\documentclass{amsart}

\usepackage{amssymb,amsmath}
\usepackage{verbatim}

\usepackage{tikz}
\usepackage{pgfplots}

\usepackage[section]{algorithm}
\usepackage{algpseudocode}

\usepackage{cleveref}

\newtheorem{theorem}{Theorem}[section]
\newtheorem{proposition}[theorem]{Proposition}

\newtheorem{definition}[theorem]{Definition}

\newtheorem{remark}[theorem]{Remark}

\newtheorem{example}[theorem]{Example}

\def\NN{{\mathbb N}}

\def\KK{{\mathbb K}}

\def\FF{{\mathbb F}}

\def\ie{\hbox{i.e.}}

\def\Ker{{\mathrm{Ker\,}}}
\def\Im{{\mathrm{Im\,}}}

\def\Gr{Gr\"obner}

\def\cC{{\mathcal{C}}}
\def\cF{{\mathcal{F}}}
\def\cO{{\mathcal{O}}}

\def\GF{{\mathrm{GF}}}

\def\Trivium{{\sc Trivium}}
\def\Magma{{\sc Magma}}

\def\bX{{\bar{X}}}
\def\bR{{\bar{R}}}

\def\bT{{\bar{\mathrm{T}}}}

\def\T{{\mathrm{T}}}

\def\bV{{\bar{V}}}

\def\cO{{\mathcal{O}}}

\def\NRV{{\mathrm{NRV}}}

\def\Reduce{{\textsc{Reduce}}}
\def\GB{{\textsc{GrobnerBasis}}}
\def\GBElimLin{{\textsc{GBElimLin}}}
\def\ElimLin{{\textsc{ElimLin}}}
\def\MultiSolve{{\textsc{MultiSolve}}}
\def\StepSolve{{\textsc{StepSolve}}}

\begin{document}

\title[A multistep strategy for polynomial system solving $\ldots$] 
{A multistep strategy for polynomial system solving over finite fields
and a new algebraic attack on the stream cipher Trivium}

\author[R. La Scala]{Roberto La Scala$^*$}

\author[F. Pintore]{Federico Pintore$^{**}$}

\author[S.K. Tiwari]{Sharwan K. Tiwari$^{\dagger}$}

\author[A. Visconti]{Andrea Visconti$^{\ddagger}$}

\address{$^*$ Dipartimento di Matematica, Universit\`a degli Studi di Bari
``Aldo Moro'', Via Orabona 4, 70125 Bari, Italy}
\email{roberto.lascala@uniba.it}

\address{$^{**}$ Dipartimento di Matematica, Universit\`a degli Studi di Trento,
Via Sommarive 14, 38123 Povo, Italy}
\email{federico.pintore@unitn.it}

\address{$^{\dagger}$ Cryptography Research Centre, Technology Innovation Institute,
Abu Dhabi, United Arab Emirates}
\email{sharwan.tiwari@tii.ae}

\address{$^{\ddagger}$ CLUB -- Cryptography and Coding Theory Group,
Dipartimento di Informatica, Universit\`a degli Studi di Milano, 
Via Celoria 18, 20133 Milano, Italy}
\email{andrea.visconti@unimi.it}

\thanks{
The first author acknowledges the partial support of PNRR MUR projects ``Security and Rights
in the CyberSpace'', Grant ref.~CUP H93C22000620001, Code PE00000014, Spoke 3,
and ``National Center for HPC, Big Data and Quantum Computing'',
Grant ref.~CUP H93C22000450007, Code CN00000013, Spoke 10.
The same author was co-funded by Universit\`a di Bari, Horizon Europe Seeds project,
Grant ref.~PANDORA - S22 and ``Fondo acquisto e manutenzione attrezzature per la ricerca'',
Grant ref.~DR 3191. The second author was partially funded by the ``Research
for Innovation Project'' (POR Puglia FESR-FSE 2014/2024), Grant ref.~366ABF6C and
Ripple’s University Blockchain Research Initiative. The fourth author was partially
supported by project SERICS (PE00000014) under the NRRP MUR program funded by
the EU - NextGenerationEU}



\begin{abstract}
In this paper we introduce a multistep generalization of the guess-and-determine
or hybrid strategy for solving a system of multivariate polynomial equations
over a finite field. In particular, we propose performing the exhaustive evaluation
of a subset of variables stepwise, that is, by incrementing the size of such subset
each time that an evaluation leads to a polynomial system which is possibly unfeasible
to solve. The decision about which evaluation to extend is based on
a preprocessing consisting in computing an incomplete \Gr\ basis after the current
evaluation, which possibly generates linear polynomials that are used to eliminate
further variables. If the number of remaining variables in the system is deemed
still too high, the evaluation is extended and the preprocessing is iterated.
Otherwise, we solve the system by a complete \Gr\ basis computation.

Having in mind cryptanalytic applications, we present an implementation of this
strategy in an algorithm called \MultiSolve\ which is designed for polynomial systems
having at most one solution. We prove explicit formulas for its complexity
which are based on probability distributions that can be easily estimated
by performing the proposed preprocessing on a testset of evaluations for different
subsets of variables. We prove that an optimal complexity of \MultiSolve\
is achieved by using a full multistep strategy with a maximum number of steps
and in turn the standard guess-and-determine strategy, which essentially is
a strategy consisting of a single step, is the worst choice.
Finally, we extensively study the behaviour of \MultiSolve\ when performing
an algebraic attack on the well-known stream cipher \Trivium.
\end{abstract}

\maketitle

\begin{quotation}
{\em Per Maria Jos\'e, in memoriam.}
\end{quotation}


\section{Introduction}

One of the most challenging and useful tasks in Computational Algebra is solving
a non-linear system of multivariate polynomial equations over a finite field.
Applications range from Discrete Logarithm Problem \cite{APS} to SAT Problem \cite{Ma},
from the computation of Error Locator Polynomials \cite{OS} to the study of classical
solutions of Quantum Designs \cite{RM}. The problem of solving a Multivariate
Polynomial system over a finite field is called ``MP problem'' and it is known
to be NP-hard (see, for instance, \cite{NPMP}). Therefore, determining
the worst-case complexity of this problem has become of great relevance
within Complexity Theory. Furthermore, the assumed difficulty of the MP problem
has been exploited as the security backbone for numerous cryptographic schemes,
motivating further research on polynomial system solving algorithms for both
cryptanalytic and design purposes.

Among symbolic algorithms to solve the MP problem, \Gr\ bases generally represent
one of the most effective options. The current best-known methods include Faug\`ere's
$F_4,F_5$ algorithms \cite{F4,F5}, as well as Courtois et al.~XL algorithm \cite{CKPS}.
Both these approaches essentially rely on performing Gaussian elimination over a
``Macaulay matrix'' which is a natural way of representing polynomials as vectors
with respect to a monomial basis. Even in accordance with the classical Buchberger's
algorithm (see, for instance, \cite{AL}), the complexity of such methods is then affected
by the size of the largest Macaulay matrix involved in the computations, which depends
on a parameter known as the ``solving degree'' of the polynomial system (see, for instance,
\cite{CaGo}). The difficulty in precisely determining this degree in advance has led
to the introduction of the concept of ``semi-regular sequences'' \cite{BFSY,BFSS} which have
theoretically-computable solving degree. Alternative parameters to analyze the complexity
of \Gr\ bases algorithms have been also considered and we refer the interested reader
to \cite{CaGo} for an overview on these parameters which form a still active area of research.

It should be mentioned that there are other effective symbolic methods available
for the MP problem such as characteristic (triangular) sets \cite{Wu} and involutive
bases \cite{GB}. Among non-symbolic solvers for the binary field $\GF(2)$, we have
well-developed and widely used algorithms such as SAT solvers, besides new promising
methods as Quantum Annealing which are about to be fully exploited for the
MP problem \cite{TII}. Nevertheless, providing accurate estimates for the complexity
of solving a non-linear system over a finite field remains a problem largely open.
In the present paper we propose a reliable statistical way to estimate this
complexity when the polynomial system has at most one solution - a quite common
case in cryptography.

When solving an instance of the MP problem, an upper bound for the complexity is
clearly provided by the exhaustive evaluation of all variables. In real-world
applications, this approach is generally unfeasible because of a large number
of variables, as well as solving the given system by a single instance of
any existing solver. A standard approach is therefore to try a ``guess-and-determine''
or ``hybrid'' strategy which consists in combining the exhaustive evaluation
of a subset of variables with the solving of all resulting polynomial systems.
This sort of divide-and-conquer strategy generally reduces an unfeasible MP problem
to many feasible instances of it and hence its complexity is essentially
the number of such instances. Different applications of this strategy in the
cryptography context appear, for instance, in the papers \cite{BFP1,CKPS,HL,LSPTV,LST}.
In all these applications, one assumes that for a suitably large number of evaluated
variables the corresponding systems can be all solved in a reasonable time.
Let us call this standard strategy a ``one-step strategy''. Note that, in order to
enhance the impact of the evaluations in reducing the number of variables,
a kind of interreduction of the \Gr\ bases theory is usually performed after
the evaluation which can give rise to some linear polynomials that are used
to eliminate further variables. This trick was introduced for the first time
in cryptanalysis by Courtois et al.~\cite{CB} as the procedure \ElimLin. A main drawback
of a one-step strategy is that the number of required variables such that
all polynomial systems corresponding to their evaluations are feasibly solvable
may be quite large leading to a huge exponential number of solving instances.

In order to reduce the total number of polynomial systems to be solved, in this paper
we propose a ``multistep strategy'' where one starts with a limited number of evaluated
variables which is increased only if the number of remaining variables after performing
an evaluation and an incomplete \Gr\ basis computation stopped at some chosen degree $D$
(including elimination via the linear polynomials obtained) is strictly greater than
a given bound $B$. Such algorithm terminates because there is some final step where
we have enough evaluated variables so that the number of remaining variables is always
smaller or equal than $B$. The multistep strategy we propose is extremely flexible
and optimizable by tuning the pararameters $D$ and $B$ and the considered set
of evaluated variables. Moreover, we are able to prove that a multistep strategy
with maximum number of steps, that is, where variables are added for evaluation
one by one, is optimal with respect to the total number of systems to be solved.
In other words, the standard one-step strategy is the worst one with respect to this
number.

As a proof-of-concept, we have implemented an algebraic attack on the well-known stream
cipher \Trivium\ by means of the proposed multistep strategy. This cipher was proposed by
De Canni\`ere and Preneel in 2005 as a submission to the European project eSTREAM \cite{DCP}.
Indeed, \Trivium\ was one of the winners of the competition for the category of
hardware-oriented ciphers. The register of this cipher is a 288-bit string and
the initial state consists of 80 bits as a private key (i.e.~the secret seed), further
80 bits are the initial vector and the remaining bits are constant. \Trivium\
is designed to generate up to $2^{64}$ bits of keystream where the first keystream bit
is produced only after $4\cdot 288 = 1152$ warm-up updates of the initial state. Despite its
simple and elegant design, \Trivium\ has so far brilliantly resisted all cryptanalytic
attacks, none of which has a better complexity than brute force over the private key
only, that is, $2^{80}$. 

Amongst the algebraic cryptanalysis of \Trivium, one finds the attack of Raddum \cite{Rad}
which aims at recovering an internal state rather than the initial one by applying
techniques from graph theory to solve a system of 954 multivariate polynomial equations
in 954 variables, obtained from the knowledge of 288 keystream bits. The complexity
of this attack is estimated to be $2^{164}$ and it was further analysed in \cite{BorKnuMat}
by using Simulated Annealing.

Another algebraic attack to recover an internal state was proposed in \cite{HL}.
This attack exploits a standard one-step strategy which, for each evaluation
of 115 variables, runs a kind of \ElimLin\ procedure to further reduce the number
of variables. The obtained systems for a testset of different keys and evaluations have
at most 33 remaining variables and they are all solvable with a variant of the characteristic
set method for finite fields, called MFCS \cite{GaoHua}. In order to reduce the degree of the
equations, this attack uses only 190 bits from the keystream allowing possible spurious
solutions. 

Among other types of attacks, it is worth noting the one proposed by Maximov and Biryukov
\cite{MaxBir} which guesses the value of some specific state bits (or the products of
state bits) leading, in some cases, to solve of a system of linear equations rather than
a system of quadratic equations. This algebraic attack hence consists in solving approximately
$2^{83.5}$ sparse linear systems consisting of 192 equations. Due to sparsity, the complexity
for solving each of these systems is about $2^{16}$. Note that for the attack to be successful,
a rather prohibitive string of $2^{61.5}$ keystream bits is required.

Finally, for the differential cryptanalysis of \Trivium\ we mention the cube attack
proposed by Dinur and Shamir in \cite{Cube}, where the adversary determines special
polynomials, called superpolies, whose variables involve key bits.
One computes superpolies by summing over a set of initial vectors which is called
a ``cube''. In fact, in this attack the stream cipher is essentially treated as
a black-box that grants the opponent access to IVs and corresponding keystreams. Many
variations and improvements of such differential attack have targeted reduced variants
of \Trivium\ (see, for instance, \cite{Ped} for an overview). By reduced we mean any
modification of \Trivium\ which considers either a reduced number $N$ of warm-up updates
(with the resulting scheme usually named $N$-round \Trivium) or a smaller register
with simplified update equations. Among the most recent results of this line
of research. it is worth mentioning \cite{HeHuPrWa} where the superpolies are obtained
for up to $848$-round \Trivium.

The present paper is organized as follows. In Section 2 we present basic facts and
techniques for polynomial system solving over any finite field $\FF$. In particular,
we recall a polynomial bound for the solving degree of such polynomial systems
and we derive a single exponential asymptotic complexity formula for the computation
of a DegRevLex-\Gr\ basis of an ideal containing the field equations of $\FF$.
The exponent of this formula involves the square of the number of variables
and the cardinality of the field $\FF$. Finally, we prove an elimination theorem
for polynomial systems containing a set of explicit equations. In Section 3, we introduce
the \MultiSolve\ algorithm able to compute the $\FF$-solution set $V_\FF(J)$
of a polynomial ideal $J$ such that $V_\FF(J)\leq 1$. This algorithm implements
the multistep strategy we propose and we discuss how its parameters can modify
the behaviour of the algorithm. In Section 4 we study the complexity of the algorithm
\MultiSolve\ in terms of the number of calls to its core subroutines \GBElimLin\
(generation of linear polynomials and elimination by means of them) and \GB.
These numbers are obtained by explicit formulas that include some probability
distributions which can be reliably estimated by applying \GBElimLin\ on suitably
large testsets for different numbers of evaluated variables. We prove that a minimum
number of calls to \GB\ (or other polynomial system solvers) is obtained
for a full multistep strategy, that is, a strategy which consists in adding
a single variable to the evaluation set. In Section 5 we recall the notion
of difference stream cipher \cite{LSPTV,LST} as a suitable formalization of stream
ciphers based on feedback shift registers and in Section 6 we discuss possible algebraic
attacks on such ciphers. In Section 7 we describe the stream cipher \Trivium\
as a difference stream cipher and we compute the inverse of its state transition map
which allows an internal state attack. In Section 8 we present an extensive experimental
study of the complexity of performing such an algebraic attack on \Trivium\ by means
of the algorithm \MultiSolve. We achieve an average complexity of $2^{106.2}$
calls to \GB\ solver where all bases can be easily computed. This improves
any previous algebraic attack based on the knowledge of a reasonably short
keystream, such as the attack in \cite{HL} which requires solving $2^{115}$
polynomial systems. Confirming the security of \Trivium, our internal state
attack is still worse than brute force over its 80-bit private key.
Nevertheless, we are optimistic that ongoing enhancements of the polynomial
system solvers, coupled with a more refined tuning for \Trivium\ of the parameters
of the \MultiSolve\ algorithm, could lead to a further narrowing of the gap.
Some other conclusions are finally drawn in Section 9.

\section{Solving polynomial systems over finite fields}

In this section we briefly review some basic results about solving a polynomial
system with coefficients and solutions over a finite field, having a low number of such
solutions. This situation is quite natural in algebraic attacks on cryptosystems where
the key is generally determined in a unique way by the given data. We start fixing
some notations. Let $\FF = \GF(q)$ be a finite field and consider the polynomial
algebra $R = \FF[x_1,\dots,x_n]$ and the ideal
$L = \langle x_1^q - x_1, \ldots, x_n^q - x_n \rangle\subset R$. If $J\subset R$ is
an ideal and $\bar{\FF}$ is the algebraic closure of the field $\FF$, we denote
\[
V(J) = \{(a_1,\ldots,a_n)\in\bar{\FF}^n\mid f(a_1,\ldots,a_n) = 0\
\mbox{for all}\ f\in J\}.
\]
and $V_\FF(J) = V(J)\cap \FF^n$. 

The Nullstellensatz over finite fields (see, for instance, \cite{Gh}) implies immediately
the following result.

\begin{proposition}
\label{nullstell}
Let $J\subset R$ be an ideal. We have that $V(L) = \FF^n$ and $V_\FF(J) = V(J + L)$ where
$J + L$ is a radical ideal of $R$.
\end{proposition}

Since $V(L) = \FF^n$, the generators of the ideal $L$ are called {\em field equations}.
An immediate consequence of the above result and the Nullstellsatz is the following result
which provides a method to solve a polynomial system with at most one solution over
the base field $\FF$.

\begin{proposition}
\label{uniqueGB}
Let $J\subset R$ be an ideal such that $\# V_\FF(J)\leq 1$.
Then, the (reduced) universal \Gr\ basis $G$ of the ideal $J + L$, that is, its \Gr\ basis
with respect to any monomial ordering of $R$ is
\[
G = 
\left\{
\begin{array}{cl}
\{x_1 - a_1,\ldots,x_n - a_n\} & \mbox{if}\ V_\FF(J) = \{(a_1,\ldots,a_n)\}, \\
\{1\} & \mbox{otherwise}.
\end{array}
\right.
\]
\end{proposition}

Indeed, as a consequence of Proposition \ref{uniqueGB}, we can choose the most efficient
monomial orderings as DegRevLex to solve a polynomial system with a single or no solutions.
Moreover, when $V_\FF(J)$ consists of few solutions, note that the cost
for obtaining them is again essentially that of computing a DegRevLex-\Gr\ basis of $J + L$.
In fact, for computing $V_\FF(J) = V(J + L)$ one needs to convert such a basis into
a Lex-\Gr\ basis by means of the FGLM-algorithm \cite{FGLM} which has complexity $\cO(n k^3)$
where $k = \# V_\FF(J) = \dim_\FF R/(J + L)$. If the integer $k$ is small, such complexity
is dominated by the cost of computing the DegRevLex-\Gr\ basis.

A precise estimation of such cost is generally difficult but it is known that the worst case
for any monomial ordering is doubly exponential in the number of variables \cite{Du}.
If $\FF = \GF(q)$ and hence $V_\FF(J) = V(J + L)$ is a finite set, algorithms for computing
DegRevLex-\Gr\ bases using linear algebra have a complexity formula (see, for instance, \cite{BFSY,BFP2})
of the form
\begin{equation}
\label{GBcomp}
\cO\bigg(\binom{n + d_s}{d_s}^\omega\bigg)
\end{equation}
where $2 < \omega\leq 3$ is the linear algebra exponent, that is, $\cO(n^\omega)$ is
the complexity for solving a linear system in $n$ variables and $d_s$ is the {\em solving degree},
that is, the highest degree of the S-polynomials involved in a complete \Gr\ basis computation
(see, for instance, \cite{CaGo}).

Note that $\binom{n + d_s}{d_s}$ is the number of monomials in $n$ variables of degree $\leq d_s$.
It refers to the number of columns of the so-called ``Macaulay matrices'' where polynomials are
viewed as vectors of their coefficients with respect to the monomial basis.

If $d_s$ is constant with respect to the number of variables $n$, one obtains clearly
a polynomial complexity. If an ideal contains the field equations of $\FF$, the solving degree $d_s$
is bounded by a quadratic function of $n$, as established in the following consequence
of the Macaulay bound (see, for instance, Theorem 11 in \cite{CaGo}).

\begin{proposition}
Let $J = \langle f_1,\ldots,f_m \rangle$ be an ideal of $R = \FF[x_1,\ldots,x_n]$ $(\FF = \GF(q))$
and consider $L = \langle x_1^q - x_1, \ldots, x_n^q - x_n \rangle\subset R$. Put $d =
\max\{d_1,\ldots,d_m,q\}$ where $d_i = \deg(f_i)$. Then, the solving degree $d_s$ for computing
a DegRevLex-\Gr\ basis of $J + L$ is bounded as follows
\begin{equation}
d_s \leq (n + 1)(d - 1) + 1.
\end{equation}
By assuming that the generators of $J$ are in normal form modulo $L$, it holds that
$d\leq n(q-1)$ $(n\geq 2)$ and one has also the following bound
\begin{equation}
\label{macbnd}
d_s\leq n^2(q - 1) + n(q - 2).
\end{equation}
\end{proposition}

We remark that, in many concrete instances, the above bounds are not at all tight. However,
they are useful to show that a DegRevLex-\Gr\ basis computation for an ideal containing
the field equations of a finite field has a single exponential complexity (see also Proposition 4
in \cite{FJ} and Theorem 5.2 in \cite{HL}).

\begin{proposition}
Let $n\geq \max\{2,q-1\}$. A $($linear algebra$)$ algorithm for computing a DegRevLex-\Gr\ basis
of an ideal $J + L\subset R$ satisfies the following asymptotic complexity
\begin{equation}
\label{GBexpcompl}
\cO(k^{k \omega})\ \mbox{where}\ k = n^2q.
\end{equation}
\end{proposition}

\begin{proof}
We start by applying the Stirling's formula to the factorials occurring in the binomial complexity
(\ref{GBcomp}). Recall this formula is
\[
\sqrt{2\pi n}\cdot e^{\frac{1}{12n + 1}}\cdot (n/e)^n \leq n! \leq
\sqrt{2\pi n}\cdot e^{\frac{1}{12n}}\cdot (n/e)^n.
\]
By applying the above estimate also to $(n + d_s)!$ and $d_s!$, we obtain
\[
\binom{n + d_s}{d_s} =
\frac{(n + d_s)!}{n! d_s!} \leq \sqrt{\frac{n+d_s}{2\pi n d_s}}\cdot e^{\frac{1}{12(n + d_s) + 1}
- \frac{1}{12n + 1} - \frac{1}{12d_s + 1}} \cdot \frac{(n+d_s)^{n+d_s}}{n^nd_s^{d_s}}.
\]
The first factor $((n+d_s)/(2\pi n d_s))^{1/2}$ is a non-increasing function bounded above by 1,
while the second factor $e^{\frac{1}{12(n+d_s)}-\frac{1}{12n+1}-\frac{1}{12d_s+1}}$ is bounded above
by $e$. Hence, we further obtain
\[
\binom{n + d_s}{d_s}\leq e\frac{(n+d_s)^{n+d_s}}{n^nd_s^{d_s}} \leq e (n+d_s)^{n+d_s}.
\]
Let us consider now the bound (\ref{macbnd}) for the solving degree. We obtain
\[
n + d_s\leq n + n^2(q - 1) + n(q - 2) = (n^2 + n)(q - 1).
\]
By observing that $(n^2 + n)(q - 1)\leq n^2q$ when $n\geq q-1$, we conclude
\[
\binom{n + d_s}{d_s}\leq e (n+d_s)^{n+d_s} \leq e (n^2q)^{n^2q}.
\]
\end{proof}

Note that better estimates of the solving degree $d_s$ are used within cryptography
when the generators of the ideal $J = \langle f_1,\ldots,f_m \rangle$ are considered
randomly generated semi-regular sequences (see, for instance, \cite{BFSS,BFP2}).
By means of such estimates and the complexity (\ref{GBcomp}) for a
DegRevLex-\Gr\ basis, an approximation of the computational effort for obtaining
the solution set $V_\FF(J)$ is provided and used to assess cryptosystem security.

\medskip
In the following sections, we make use of the idea that having linear equations
in a polynomial system, or more generally explicit equations, provides a method
to eliminate variables while preserving solutions. Since solving complexity heavily
depends on the number of variables, such an elimination can significantly improve
the search for solutions. The following results provide a formalization of this
idea and its effectiveness. Till the end of this section, we assume that $\FF$ is
any base field.

Let $R = \FF[x_1,\ldots,x_n], R' = \FF[x_1,\ldots,x_m]$ be polynomial algebras
with $n\geq m$. Let $f_1,\ldots,f_{n-m}\in R'$ and consider the algebra
homomorphism $\psi:R\to R'$ such that
\[
x_i\mapsto
\left\{
\begin{array}{cl}
x_i & \mbox{if}\ i\leq m, \\
f_{i - m} & \mbox{otherwise}.
\end{array}
\right.
\]
It holds clearly that $\Im\psi = R'$ and $I = \langle x_{m+1} - f_1,\ldots,x_n - f_{n-m} \rangle
\subset \Ker\psi$. By a rewriting process over the algebra $R$ modulo the ideal $I$,
we have that $\Ker\psi = I$. In other words, the quotient algebra $R/I$ is isomorphic to
the subalgebra $R'\subset R$ by the mapping $f + I\mapsto \psi(f)$.

\begin{proposition}
\label{elim0}
Let $I\subset J\subset R$ be an ideal. We have that the elimination ideal
$J\cap R'$ is equal to $\psi(J)$.
\end{proposition}

\begin{proof}
Since $\psi$ is the identity map on $R'$, we have that $J\cap R' =
\psi(J\cap R')\subset \psi(J)$. Let $f\in J$ and consider
$f' = \psi(f)\in\psi(J)\subset R'$. We want to prove that $f'\in J$.
Since $f'\in R'$ we have that $\psi(f') = f' = \psi(f)$ and hence
$f - f'\in I\subset J$. Because $f\in J$ we conclude that $f'\in J$
and therefore $\psi(J)\subset J\cap R'$.
\end{proof}

Consider the map $\varphi:\FF^m\to\FF^n$ such that
\[
(a_1,\ldots,a_m)\mapsto
(a_1,\ldots,a_m,f_1(a_1,\ldots,a_m),\ldots,f_{n-m}(a_1,\ldots,a_m)).
\]
Note that $\varphi$ is the injective polynomial map corresponding to the
surjective algebra homomorphism $\psi$, that is, $\psi(g)(a_1,\ldots,a_m) = 
g(\varphi(a_1,\ldots,a_m))$ for all $g\in R$ and $(a_1,\dots,a_m)\in\FF^m$.
Consider now the solution set $V(J) = \{(a_1,\ldots,a_n)\in\bar{\FF}^n\mid
f(a_1,\ldots,a_n) = 0\ \mbox{for all}\ f\in J\}$ where $\bar{\FF}$ is the
algebraic closure of the field $\FF$.

\begin{proposition}
\label{elim}
Let $I\subset J\subset R$ be an ideal. We have that $V(J) = \varphi(V(J\cap R')) =
\varphi(V(\psi(J)))$.
\end{proposition}

\begin{proof}
Let $\alpha = (a_1,\ldots,a_n)\in V(J)$, that is, $f(\alpha) = 0$ for all $f\in J$.
If $\alpha' = (a_1,\ldots,a_m)$, it is clear that $f'(\alpha') = 0$ for all
$f'\in J\cap R'$, that is, one has that $\alpha'\in V(J\cap R')$. Since
$I = \langle x_{m+1} - f_1,\ldots,x_n - f_{n-m} \rangle\subset J$, it holds that
$a_{m+1} = f_1(\alpha'),\ldots,a_n = f_{n-m}(\alpha')$ and we conclude that
$V(J)\subset \varphi(V(J\cap R'))$.

Assume now that $\alpha' = (a_1,\ldots,a_m)\in V(J\cap R')$ and consider
$\alpha = (a_1,\ldots,a_n) = \varphi(\alpha')$. Let $f\in J$ and denote
$f' = \psi(f)\in \psi(J)$. By Proposition \ref{elim0}, we have that
$f'\in J\cap R'$ and since $\psi(f') = f' = \psi(f)$ it holds that
$f - f'\in I$. This implies that $f'(\alpha') = 0$ and $f = f' + g$ with
$g\in I$. Since $\alpha = \varphi(\alpha')$ we have that $f(\alpha) = 0$.
In fact, it holds that $g(\alpha) = g(\varphi(\alpha'))
= \psi(g)(\alpha') = 0$ because $g\in I = \Ker\psi$. We conclude that
$\varphi(V(J\cap R'))\subset V(J)$.
Finally, because $J\cap R' = \psi(J)$ by Proposition \ref{elim0},
it holds that $V(J) = \varphi(V(J\cap R')) = \varphi(V(\psi(J)))$.
\end{proof}

Informally, the above proposition states that the solutions of a polynomial system
containing some explicit equations $x_{m+i} = f_i$ can be all computed by extending
the solutions of the polynomial system obtained by eliminating the variables $x_{m+i}$.
Observe that if $\deg(f_i) > 1$, the advantage of eliminating variables from equations
is partially reduced by the growth of their degrees which may affect the solving
degree of the system. It is not the case when we eliminate via linear equations,
which is therefore the preferable situation. We elaborate and explore this viewpoint
in the next section.

\section{Solving by a multistep strategy}

The complexities in the previous section show that computing a \Gr\ basis for a quite
large number of variables may be a hard task. The amount of computation does not
change significantly if we substitute \Gr\ bases with other solvers like SAT, BDD, etc.~(see,
for instance, \cite{LSPTV}). Indeed, it is well-know that the problem of solving multivariate
polynomial systems over finite fields is an NP-hard one \cite{NPMP}. A useful technique consists
therefore in substituting an unfeasible \Gr\ basis computation with multiple feasible
computations for polynomial systems obtained from the original one by evaluating a subset
of variables over the base finite field. Of course, the exponential complexity is hidden here
in such exhaustive evaluation. This approach is usually called a {\em guess-and-determine}
or {\em hybrid strategy}, where the latter name is especially used for polynomial systems
that are randomly generated \cite{BFSS,BFP2}. In this case, by evaluating some variables
one essentially obtains again a random system whose complexity can be estimated
in a theoretical way. The computation of the total complexity may show that in some cases
the hybrid approach is faster than a single \Gr\ basis computation \cite{BFP1,BFP2}.

In this section, we essentially introduce a variant of this technique where given
polynomial systems are generally not random and hence the evaluation of a subset
of their variables may lead to systems with different behaviours. In particular,
the main property we analyze for the resulting systems is the number of linear
equations that can be obtained by an incomplete \Gr\ basis computation which is stopped
at some chosen degree $D$. In fact, the presence of these linear equations allows
the elimination of a further amount of variables, without increasing the degrees of the
system equations, with respect to the set of evaluated variables. The number of remaining
variables, $\NRV$ in short, in these systems gives us an idea of the complexity
we will face in solving them completely. Given a bound $0\leq B\leq n$, we decide
to compute a complete \Gr\ basis only if $\NRV\leq B$. Otherwise, we evaluate some
additional amount of variables and we compute again the $\NRV$ which is then compared to $B$.
Proceeding in this way, we are able to solve the given polynomial system with \Gr\ bases
over a number of variables which is always bounded by $B$. Assuming that each of these
\Gr\ basis computations can be performed in a reasonable time, the complexity is essentially
the number of them. We will show how to estimate in practice this number by means
of explicit formulas and simple statistics and how to optimize it by choosing appropriate
steps when adding new variables to evaluate. We call this approach a {\em multistep
strategy} because of such steps and because \Gr\ bases are also computed stepwise.
After the above high-level description of the strategy, in the following we formally
describe it.

Let $R = \FF[x_1,\ldots,x_n]$ be a polynomial algebra over the finite field $\FF = \GF(q)$.
Consider an ideal $J = \langle f_1,\ldots,f_m \rangle\subset R$ and assume that
$\# V_\FF(J)\leq 1$. Denote as usual $L = \langle x_1^q - x_1, \ldots, x_n^q - x_n \rangle$
and put $H = \{f_i\}\cup\{x_i^q - x_i\}$ which is a generating set of the ideal $J + L$.
Let $D > 0$ be an integer and denote by $\GB(H, D)$ the algorithm performing an incomplete
\Gr\ basis computation which is stopped when all S-polynomials of degree $\leq D$
have been consider at a current step. This algorithm corresponds to put
a bound on the size of the Macaulay matrix when using linear algebra methods.
If $D$ is suitably large, that is, it is a solving degree we have that $G = \GB(H, D)$
is a complete \Gr\ basis of $J + L$. In the considered case that $\#V_\FF(J)\leq 1$,
the complete \Gr\ basis $G$ contains only linear polynomials, namely $G = \{1\}$ or
$G = \{x_1 - a_1,\ldots,x_n - a_n\}$ (see Proposition \ref{uniqueGB}). Of course,
computing the complete \Gr\ basis is generally expensive because of the complexity (\ref{GBcomp})
and to know linear polynomials belonging to the ideal $J + L$ is useful for eliminating
variables, without increasing degrees, according to Proposition \ref{elim}. It is reasonable
therefore to compute $G = \GB(H, D)$, for a quite low degree $D$ in order to have a fast
computation, and use the linear polynomials that we possibly found in $G$ to eliminate variables
from the generators of $J + L$. If $G = \{1\}$ or $G = \{x_1 - a_1,\ldots,x_n - a_n\}$,
no further computation is needed. Otherwise, such a preprocessing will generally ease
the task of computing $V_\FF(J)$. We formalize it as the Algorithm 3.1.

\suppressfloats[b]
\floatname{algorithm}{Algorithm}
\begin{algorithm}
\caption{\GBElimLin}
\begin{tabular}{l@{\ }l}
\hspace{-16pt} {\bf Input}: & a \hfill generating \hfill set \hfill $H$ \hfill of \hfill an \hfill ideal \hfill
                $J + L$ \hfill ($\#V_\FF(J)\leq 1$) \hfill and \hfill an \hfill integer \\
                & $D > 0$; \\
\hspace{-16pt} {\bf Output}: & either \hfill a \hfill \Gr\ \hfill basis \hfill of \hfill $J + L$ \hfill or \hfill
                a \hfill generating \hfill set \hfill of \hfill an \hfill ideal \\
                & obtained \hfill by \hfill eliminating \hfill some \hfill variables, \hfill via \hfill
		linear \hfill equations, \hfill from \\
		& $J + L$. \\
\end{tabular}
\begin{algorithmic}[0]
\State $G:= \GB(H, D)$;
\If{$\max\{\deg(g)\mid g\in G\}\leq 1$}
\State \Return G;
\EndIf;
\State $G_1:= \{g\in G\mid \deg(g) = 1\}$;
\State $G_2:= G\setminus G_1$;
\State $S:= \{x_i^q - x_i\mid x_i\ \mbox{occurs in}\ G\}$;
\State $G:= \Reduce(G_2, G_1 \cup S)$;
\State $S:= \{x_i^q - x_i\mid x_i\ \mbox{occurs in}\ G\}$;
\State $G:= G\cup S$;
\State \Return $G$;
\end{algorithmic}
\end{algorithm}

In the algorithm \GBElimLin\ we assume that the linear polynomials of the set $G_1$ are
completely interreduced, that is, they are obtained in reduced echelon form in the Macaulay
matrix. The procedure \Reduce\ is the complete reduction algorithm of \Gr\ bases theory
(see, for instance, \cite{AL}). If it holds that $\max\{\deg(g)\mid g\in G\}\leq 1$, namely
$G = \{1\}$ or $G = \{x_1 - a_1,\ldots,x_n - a_n\}$, then $\GBElimLin$ provides immediately
$V_\FF(J) = V(J + L)$ because it happens that $D$ is a solving degree. Otherwise, we have
that the solution set of the ideal $J' + L' = \langle G\rangle$ ($G = \GBElimLin(H,D)$)
is immediately related to $V(J + L)$ according to Proposition \ref{elim}. In this case,
we define $\NRV > 0$ as the number of variables occurring in $G$.
Note that in order to compute all coordinates of a possible solution, we should output
also the set of linear equations $G_1$ from the procedure $\GBElimLin$. We refrain from
doing so because in the solving algorithm that we are about to present, we need this set
for a single call at $\GBElimLin$.

Since the incomplete \Gr\ basis computation $\GB(H, D)$ is performed in \GBElimLin\
and it is assumed that the complete \Gr\ basis of $H$ is full linear, the set of linear
generators $G_1$ is usually non-empty. Consequently, the number of remaining variables $\NRV$
is typically less than the number of variables in $H$ due to the elimination process
$\Reduce(G_2, G_1 \cup S)$.  Note that such elimination of variables by linear equations
preserves the degrees of the non-linear generators in $G_2$ or may even reduce them,
owing to the field equations in $S$. Therefore, the task of computing the solution set
$V_\FF(J')$ is easier than computing $V_\FF(J)$. This simplification is obtained at the price
of performing \GBElimLin\ which may be quite cheap whenever $D$ is a low degree.
This occurs, for instance, when $D = \max\{\deg(g)\mid g\in H\}$ and $q$ is a small integer.
Note that in this case, the procedure $\GB(H, D)$ is essentially a kind of interreduction
of \Gr\ bases theory and the procedure \GBElimLin\ practically coincides with linear algebra
based algorithm \ElimLin\ \cite{Ba,CB}.

Fix now a bound $0\leq B\leq n$ and assume that a complete \Gr\ basis over a number of variables
less than or equal to $B$ can be computed in a reasonable time. If the number of variables $n$
is quite large as in polynomial systems arising in cryptography, we cannot immediately apply
$\GBElimLin(H, D)$. Then, we add to the generating set $H$ of $J + L$ a set of evaluations
of some suitable amount of variables, say
\[
E = \{ x_1 - a_1,\ldots,x_k - a_k \}\ (a_i\in\FF)
\]
and compute $G = \GBElimLin(H\cup E, D)$. As previously noted, we work under the assumption that
the complete \Gr\ basis of $H$ and hence of $H\cup E$ is a linear one. More precisely, if $G'$ is the
complete \Gr\ basis of $H\cup E$, then either $G' = \{1\}$ or $G' = \{x_1 - a_1,\ldots,x_n - a_n\}$.
Although $D$ is generally lower than the solving degree for $H\cup E$, the incomplete \Gr\ basis $G''$
of $H\cup E$ computed in \GBElimLin\ approaches the linear basis $G'$ and includes $E$ whenever
$G''\neq \{1\}$. In other words, some new linear polynomials are usually computed in addition
to those in $E$, that is, $\NRV$ is typically less than $n - k$. This is especially evident
for non-random systems, such as those arising in the cryptanalysis of specific (stream) ciphers.

This suggests to moderate the number $k$ which determines the exponential complexity $q^k$
corresponding to the exhaustive evaluation of $k$ variables over the field $\FF = \GF(q)$. If $\NRV\leq B$,
we proceed with the computation of $\GB(G)$. Otherwise, we extend the evaluation set $E$ to
\[
E' = \{ x_1 - a_1,\ldots,x_l - a_l \}\ (k < l)
\]
and compute $G' = \GBElimLin(H\cup E', D)$. By iterating the above steps,
the solution set $V_\FF(J)$ ($\# V_\FF(J)\leq 1$) is obtained by adding each time
a small amount of variables to evaluate, till either all considered evaluations
are exhausted or the output of \GBElimLin\ or \GB\ ($\NRV\leq B$) is a \Gr\ basis $G$
such that $\max\{\deg(g)\mid g\in G\} = 1$. The Algorithm 3.3 implement this multistep
strategy together with a suitable management of the exhaustive evaluation of different
amounts of variables. A fundamental subroutine is the Algorithm 3.2.

\suppressfloats[b]
\floatname{algorithm}{Algorithm}
\begin{algorithm}
\caption{\StepSolve}
\begin{tabular}{l@{\ }l}
\ {\bf Input}: & a \hfill generating \hfill set \hfill $H$ \hfill of \hfill an \hfill ideal
\hfill $J + L$ \hfill ($\#V_\FF(J)\leq 1$), \hfill a \hfill subset \hfill $W\subset\FF^k$ \\
& ($1\leq k\leq n$) and three integers $k\leq l\leq n, 0\leq B\leq n, D > 0$;  \\
\ {\bf Output}: & either \hfill a \hfill pair \hfill ({\tt solution}, $G$) \hfill where \hfill
{\tt solution} \hfill is \hfill a \hfill character \hfill string \hfill and \\
& $G\neq \{1\}$ \hfill is \hfill a \hfill \Gr\ \hfill basis \hfill of \hfill $J + L$ \hfill
or \hfill a \hfill \Gr\ \hfill basis \hfill of \hfill an \hfill ideal \\
&  obtained \hfill by \hfill eliminating \hfill some \hfill variables \hfill from \hfill $J + L$,
\hfill or \hfill the \hfill pair \\
& ({\tt wild-set}, $W'$) \hfill where \hfill {\tt wild-set} \hfill is \hfill a \hfill character
\hfill string \hfill and \hfill $W'$ \hfill is \hfill a \\
& subset of $\FF^l$. \\
\end{tabular}
\begin{algorithmic}[0]
\If{$k\neq l$}
\State $W:= W\times \FF^{l-k}$;
\EndIf;
\State $W':= \emptyset$;
\ForAll{$(a_1,\ldots,a_l)\in W$}
\State $E:= \{ x_1 - a_1, \ldots, x_l - a_l \}$;
\State $G:= \GBElimLin(H \cup E, D)$;
\If{$\max\{\deg(g) \mid g\in G\} = 1$}
\State {\bf print} $E$;
\State \Return ({\tt solution}, $G$);
\EndIf;
\If{$G\neq \{1\}$}
\State $\NRV:=$ number of variables in $G$;
\If{$\NRV\leq B$}
\State $G:= \GB(G)$;
\If{$\max\{\deg(g) \mid g\in G\} = 1$}
\State {\bf print} $E$;
\State \Return ({\tt solution}, $G$);
\EndIf;
\Else
\State $W':= W'\cup\{(a_1,\ldots,a_l)\}$;
\EndIf;
\EndIf;
\EndFor;
\State \Return ({\tt wild-set}, $W'$);
\end{algorithmic}
\end{algorithm}

Note that \StepSolve\ may output a \Gr\ basis $G\neq\{1\}$ of an ideal obtained by
eliminating some variables from the input ideal $J + L$ by means of \GBElimLin.
By Proposition \ref{elim}, one obtains the actual \Gr\ basis
$G' = \{x_1 - a_1,\ldots,x_n - a_n\}$ of $J + L$, or equivalently $V_\FF(J) =
\{(a_1,\ldots,a_n)\}$, simply by extending the coordinates in $G$ via the linear equations 
computed within the procedure $\GBElimLin((H \cup E, D)$ where
$E = \{ x_1 - a_1, \ldots, x_l - a_l \}$ is printed by \StepSolve. Such equations
are indeed necessary in this unique case.

\suppressfloats[b]
\floatname{algorithm}{Algorithm}
\begin{algorithm}
\caption{\MultiSolve}
\begin{tabular}{l@{\ }l}
\hspace{-12pt} {\bf Input}: & a \hfill generating \hfill set \hfill $H$ \hfill of \hfill an \hfill ideal
\hfill $J + L$ \hfill ($V_\FF(J)\leq 1$), \hfill a \hfill sequence \hfill of \\
& integers $1\leq k_1 < \ldots < k_r\leq n$ and two integers $0\leq B\leq n, D > 0$; \\
\hspace{-12pt} {\bf Output}: & either \hfill a \hfill \Gr\ \hfill basis \hfill of \hfill $J + L$ \hfill
or \hfill a \hfill \Gr\ \hfill basis \hfill of \hfill an \hfill ideal \\
&  obtained \hfill by \hfill eliminating \hfill some \hfill variables, \hfill via \hfill linear \hfill
equations, \hfill from \\
& the ideal $J + L$.
\end{tabular}
\begin{algorithmic}[0]
\State \ForAll{$(a_1,\ldots,a_{k_1})\in\FF^{k_1}$}
\State $W:= \{(a_1,\ldots,a_{k_1})\}$;
\State $i:= 1$;
\While{$W\neq\emptyset$}
\State $(A_1,A_2):= \StepSolve(H,W,k_i,D,B)$;
\If{$A_1 = {\tt solution}$}
\State \Return $A_2$;
\EndIf;
\State $W:= A_2$;
\State $i:= i + 1$;
\EndWhile;
\EndFor;
\State \Return $\{1\}$;
\end{algorithmic}
\end{algorithm}

The algorithm \MultiSolve\ yields the desired output on the condition that
for each $(a_1,\ldots,a_{k_1})\in\FF^{k_1}$, following the call at \StepSolve\
for the final step $k_r$, one has that $A_2$ is the empty set when $A_1 = {\tt wild-set}$.
In fact, we are assuming that $1\leq k_r\leq n$ is the least integer such that
this holds. Of course, such an integer exists because for $k_r = n$ we have
the evaluation of all variables. Note that our solving algorithm only requires
the storage of the wild-set $W$ for each iteration of \StepSolve\ and for every
guess $(a_1,\ldots,a_{k_1})$. If a solution is found, the algorithm terminates
promptly, avoiding any further useless computations.

We remark that \MultiSolve\ generalizes the standard guess-and-determine
(or hybrid) strategy in the case that we have a single iteration, namely for $r = 1$
and $k_r = k_1$. The only difference with the standard strategy is that \MultiSolve\
compute the \Gr\ bases for all evaluations $(a_1,\ldots,a_k)\in\FF^k$ ($k = k_r = k_1$)
in two steps, first executing \GBElimLin\ and then by \GB. If $\NRV(\leq B)$ is fairly
smaller than $n - k$, this approach can reduce computing times, especially for
\Gr\ bases algorithms based on Macaulay matrices and linear algebra. Indeed,
this improvement of the standard strategy already appeared in the algorithm
\ElimLin\ \cite{Ba,CB}.

We mention that in many computer algebra systems, such as \Magma, a restart
mechanism is integrated into the \GB\ algorithm based on the generation of
a sufficiently large number of linear equations. To use it in our \MultiSolve\
algorithm, however, one should have the possibility to input the parameters $B$
and $D$ in order to stop a potentially unfeasible computation whenever
the degree $D$ is reached but $\NRV > B$. In this case, there should also be
instructions that increase the number of variables that are exhaustively evaluated.
In the next section we will show indeed that the algorithm \MultiSolve\ for $r > 1$
outperforms the standard strategy by reducing the total number of \Gr\ bases
required to compute the solution set $V_\FF(J)$.

Observe that if $B = 0$ the algorithm \MultiSolve\ computes the output only
by applying \GBElimLin, that is, by \Gr\ bases computations which are stopped
at degree $D$. This is possible because for a sufficiently large final step
$k_r\leq n$, the solving degree of all obtained systems becomes lower or equal
than $D$.

Before applying \MultiSolve, it would be helpful to run \GBElimLin\ on the
generating set $H$ of the ideal $J + L$. Indeed, this is generally unfeasible
if $H$ has a large number of variables. Nevertheless, if there are given linear
polynomials in $H$, a good practice consists in initially eliminating variables
by means of these generators. We remark that beside the choice of the
integers $D$ and $B$, a fundamental issue in the optimization of the algorithm
\MultiSolve\ consists in the choice of the subsets
\[
\{x_1,\ldots,x_{k_1}\}\subset \{x_1,\ldots,x_{k_2}\} \subset \ldots
\subset \{x_1,\ldots,x_{k_r}\}.
\]
In the next section, we will show that an optimal choice for the steps $k_i$ is to put
$k_{i-1} = k_i - 1$. The main issue is therefore to single out the subset $\{x_1,\ldots,x_{k_r}\}$
which can have the most favorable impact on the total running time of \MultiSolve.
We suggest that in the search for an optimal subset of variables to evaluate, one can use
as objective function the average value of the $\NRV$ numbers obtained by \GBElimLin\ for the
evaluations of $\{x_1,\ldots,x_{k_r}\}$ on a suitably large testset $T\subset\FF^{k_r}$.
Of course, if the number of such subsets of the set of all variables is huge, such an optimization
cannot be performed by brute force. In this case, we can possibly search among a restricted
amount of subsets which appear to be good candidates due to the specific form of the given
system. Devising more effective optimization strategies remains an open problem which
we defer to future work.

\section{\MultiSolve\ complexity}

In this section we analyze the complexity of the algorithm \MultiSolve\ which
essentially corresponds to the number of calls to the subroutines \GBElimLin\
and \GB. In fact, both these procedures can be executed in a reasonable time
for suitable values of the parameters $D > 0$ and $0\leq B\leq n$.

We call {\em $k$-guess} an evaluation of the first $k$ variables of the polynomial
algebra $R = \FF[x_1,\ldots,x_n]$ ($\FF = \GF(q)$), that is, to make a guess
$(a_1,\ldots,a_k)\in\FF^k$ means that we are computing modulo the ideal generated
by the set
\[
E = \{ x_1 - a_1,\ldots, x_k - a_k \}.
\]
A guess is called {\em wild} if the procedure $\GBElimLin(H\cup E, D)$ ($H$ is a
generating set of the ideal $J + L$) outputs a set of generators $G$ such that
$\max\{\deg(g) \mid g\in G\} > 1$ and its corresponding $\NRV$ is strictly greater
than $B$. Otherwise, we say that the guess is {\em tamed}. Note that
the latter includes the case that \GBElimLin\ outputs a complete \Gr\ basis,
that is, $\max\{\deg(g) \mid g\in G\}\leq 1$ because $V_\FF(J)\leq 1$.

Denote $p_B(k)$ the probability of wild $k$-guesses in the set $\FF^k$ of all $k$-guesses,
that is, $p_B(k) q^k$ is the number of wild $k$-guesses that the algorithm \GBElimLin\
determines in the set $\FF^k$. In practice, such probability is estimated on some random testset
$T\subset \FF^k$ of reasonable large size.

The results about the complexity of \MultiSolve\ we are going to present, are all based
on the safe assumption that extending a tamed $k$-guess to an $l$-guess ($k < l$)
we always obtain again a tamed $l$-guess. This can be explained because reducing the number
of variables in an ideal by evaluation also reduces its solving degree and the incomplete
\Gr\ basis computed by \GBElimLin\ approaches the complete one which is full linear.
A first consequence is that the map $k\mapsto p_B(k)$ is a decreasing function.

\begin{proposition}
Let $1\leq k_1\leq \ldots \leq k_r\leq n$. One has that $1\geq p_B(k_1)\geq \ldots \geq p_B(k_r)\geq 0$
where $p_B(n) = 0$.
\end{proposition}

\begin{proof}
Let $1\leq k\leq l\leq n$. The number of all $l$-guesses extending wild $k$-guesses is
\[
p(k) q^k q^{l-k} = p(k) q^l.
\]
Since restricting a wild $l$-guess we always obtain a wild $k$-guess for the assumption
discussed above, we have that all wild $l$-guesses in the set $\FF^l$ are obtained by extending
some wild $k$-guess. This implies that $p(k) q^l\geq p(l) q^l$ and hence $p(k)\geq p(l)$.
Notice that some extensions of a wild $k$-guess may become tamed $l$-guesses.
In particular, it is clear that $p_B(n) = 0$ since we have the evaluation of all
variables of $R$.
\end{proof}

Starting from now, we denote
\[
k'' = \min\{1\leq k\leq n \mid p_B(k) = 0\}.
\]
Given a sequence of steps $1\leq k_1\leq \ldots \leq k_r = k''$, we want to compute
the total number of guesses that are considered in the algorithm \MultiSolve\ which coincides
with the number of calls to the subroutine \GBElimLin. We show that this number only depends
on the probabilities $p_B(k_1)\geq \ldots \geq p_B(k_{r-1})$ that can be estimated
on a testset without actually performing \MultiSolve.

\begin{proposition}
The number of guesses in the algorithm \MultiSolve\ for the steps
$1\leq k_1\leq \ldots \leq k_r = k''$ is
\[
C_1 = q^{k_1} + p_B(k_1)q^{k_2} + \ldots + p_B(k_{r-1}) q^{k_r}.
\]
\end{proposition}

\begin{proof}
To simplify notations, let $k < l < m < \ldots$ be the steps of the algorithm \MultiSolve.
In the first step $k$, we have to consider all $k$-guesses whose number is $q^k$.
In the second step $l$, we have to extend to $l$-guesses all wild $k$-guesses whose
number is $p_B(k) q^k$. The number of such extensions is therefore
\[
p_B(k) q^k q^{l-k} = p_B(k) q^l.
\]
For the third step $m$, in this subset of all $l$-guesses extending wild $k$-guesses
we have now to consider wild $l$-guesses to be extended to $m$-guesses. Since by restricting
a wild $l$-guess we always obtain a wild $k$-guess, this number is exactly the number
of all wild $l$-guesses in the set $\FF^l$ which is $p_B(l) q^l$. Then, the number of
$m$-guesses which extends wild $l$-guesses (extending wild $k$-guesses) is
\[
p_B(l) q^l q^{m-l} = p_B(l) q^m.
\]
We deduce that up to the third step, the total number of guesses considered in the algorithm
\MultiSolve\ is
\[
q^k + p_B(k)q^l + p_B(l) q^m.
\]
Iterating for all steps of the algorithm, one obtains the complexity $C_1$.
\end{proof}

Let us consider now the number of times the subroutine \GB\ is executed in the algorithm
\MultiSolve, which generally gives main contribution to its total running time. By definition,
this number is the cardinality of tamed guesses considered in the algorithm.

\begin{proposition}
The number of tamed guesses in the algorithm \MultiSolve\ is
\begin{equation*}
\begin{gathered}
C_2 = (1 - p_B(k_1))q^{k_1} + (p_B(k_1) - p_B(k_2))q^{k_2} + \ldots \\
+\, (p_B(k_{r-2}) - p_B(k_{r-1})) q^{k_{r-1}} + p_B(k_{r-1}) q^{k_r}.
\end{gathered}
\end{equation*}
\end{proposition}

\begin{proof}
As shown in the above results, the number of wild guesses in algorithm
\MultiSolve\ only depends on the probability sequence
$1\geq p_B(k_1)\geq \ldots \geq p_B(k_{r-1})\geq p_B(k_r) = 0$.
This total number is
\[
C = p_B(k_1) q^{k_1} + p_B(k_2) q^{k_2} + \ldots + p_B(k_{r-1}) q^{k_{r-1}}.
\]
We conclude that the total number of tamed guesses is $C_2 = C_1 - C$.
\end{proof}

We can ask now if there is an optimal choice for the subset
$\{k_1,\ldots,k_r\}\subset\{1,\ldots,n\}$ with respect to the main complexity
formula $C_2$. Recall that in this formula the last step is fixed as
$k_r = k'' = \min\{1\leq k\leq n \mid p_B(k) = 0\}$. In other words,
the step $k''$ is the least integer such that $\GBElimLin(H\cup E, D)$ ($H$ is
a generating set of $J + L$ and $E = \{x_1 - a_1,\ldots,x_{k''} - a_{k''}\}$)
obtains $\NRV\leq B$ for all guesses $(a_1,\ldots,a_{k''})\in\FF^{k''}$.
We have then the following result.

\begin{theorem}
\label{triumph}
The minimum value of the complexity $C_2$ is obtained for $k_i = i$ $(1\leq i\leq k'')$,
that is, for the maximum number of steps.
\end{theorem}

\begin{proof}
To prove the statement is true, it is sufficient to show that each time we add a new step
between two steps or before all steps, the sum $C_2$ decreases. We distinguish these two cases.

{\em Case 1}. Let $1\leq k < l < m\leq k''$ be three consecutive steps in the algorithm
\MultiSolve. In other words, we consider $l$ as the step that we possibly add in between
two steps. If all three steps occur, the formula $C_2$ contains the following two summands
\[
(p_B(k) - p_B(l))q^l + (p_B(l) - p_B(m))q^m.
\]
Otherwise, if step $l$ is missing, we have a single corresponding summand in $C_2$, namely
\[
(p_B(k) - p_B(m))q^m.
\]
Note that all remaining summands in $C_2$ are the same in these two subcases. Now, since
\[
(p_B(k) - p_B(m))q^m = (p_B(k) - p_B(l))q^m + (p_B(l) - p_B(m))q^m
\]
and $(p_B(k) - p_B(l))q^m\geq (p_B(k) - p_B(l))q^l$ because $p_B(k)\geq p_B(l)$ and $m > l$,
we conclude that
\[
(p_B(k) - p_B(m))q^m \geq (p_B(k) - p_B(l))q^l + (p_B(l) - p_B(m))q^m.
\]

{\em Case 2}. Let $1\leq k < l\leq k''$ and assume that either the step $k$ or $l$
is the first one. In other words, we consider $k$ as the step we possibly add before
all steps. If $k$ is the first step, the formula $C_2$ contains the summands
\[
(1 - p_B(k))q^k + (p_B(k) - p_B(l))q^l.
\]
Otherwise, if the step $l$ is the first one, that is, the step $k$ is missing, we have
a single corresponding summand in $C_2$, namely
\[
(1 - p_B(l))q^l.
\]
Similarly to case 1, one has that
\[
(1 - p_B(l))q^l = (1 - p_B(k))q^l + (p_B(k) - p_B(l))q^l
\]
where $(1 - p_B(k))q^l\geq (1 - p_B(k))q^k$ because $1\geq p_B(k)$ and $l > k$. We conclude
therefore that
\[
(1 - p_B(l))q^l\geq (1 - p_B(k))q^k + (p_B(k) - p_B(l))q^l.
\]
\end{proof}

Note that the above result shows in particular that the standard guess-and-determine strategy
which corresponds to apply \MultiSolve\ for the single step $k''$ is actually the worst one
with respect to complexity $C_2$. As for the first step, say $k'$, the above result suggests
that $k' = 1$ but in practice $k'$ is the smallest integer such that $\GBElimLin(H\cup E, D)$
($E = \{x_1 - a_1,\ldots,x_{k'} - a_{k'}\}$) can be performed in a reasonable time
for the given parameter $D > 0$ and for all $(a_1,\ldots,a_{k'})\in\FF^{k'}$.
For polynomial systems with many variables, such as the ones arising in cryptography,
the first step $1\leq k'\leq n$ will be a not so small integer.
We call {\em full multistep strategy} the one corresponding to the minimum value of $C_2$
which has the following complexity formulas
\begin{equation}
\begin{gathered}
C_1 = \sum_{k'\leq k\leq k''} p_B(k-1) q^k; \\
C_2 = \sum_{k'\leq k\leq k''} (p_B(k-1) - p_B(k)) q^k,
\end{gathered}
\end{equation}
where $p_B(k'') = 0$ and by convention $p_B(k'-1) = 1$.

Of course, each summand of the above complexity formulas should be multiplied
by suitable average timings in order to obtain running times. Indeed, we have
the following corresponding total running times
\begin{equation}
\label{runtime}
\begin{gathered}
T_1 = \sum_{k'\leq k\leq k''} \sigma(k-1,k) p_B(k-1) q^k; \\
T_2 = \sum_{k'\leq k\leq k''} \tau(k-1,k) (p_B(k-1) - p_B(k)) q^k,
\end{gathered}
\end{equation}
where $\sigma(k-1,k)$ is the average running time of \GBElimLin\ for the extensions
of wild $(k-1)$-guesses to $k$-guesses and $\tau(k-1,k)$ is the average running time
of \GB\ for the extensions of wild $(k-1)$-guesses to tamed $k$-guesses.
Note that for the first step $k'$ we have the timings $\sigma(k'-1,k')$ and
$\tau(k'-1,k')$ that, by convention, correspond to all $k'$-guesses and all
tamed $k'$-guesses, respectively. The total execution time of \MultiSolve\
is therefore $T = T_1 + T_2$.

\medskip
We remark that in our experimental studies, which we will detail in the next sections,
we observe that the probabilities $p_B(k)$ ($k'\leq k\leq k''$) appear to be a reliable
statistical data about a multivariate polynomial system. These data can be safely used
therefore to estimate the complexity of solving a polynomial system by the multistep algorithm
\MultiSolve, as well to tune it to an optimal complexity by varying its main parameters,
namely the integers $D,B$ and the subset of $k''$ variables to evaluate.
Note also that in the algorithm \MultiSolve\ one can possibly substitute \GB\ solver
with any other polynomial system solver over finite fields, such as XL solvers,
SAT solvers and so on. The only step where an (incomplete) \Gr\ basis computation
is essential in \MultiSolve\ is in the subroutine \GBElimLin\ where this is used
to generate additional linear equations starting from variable evaluations.
An alternative way may be to use the linear algebra algorithm \ElimLin\ \cite{Ba,CB}.
We also recall that the probabilities $p_B(k)$ and hence the multistep complexities
$C_1,C_2$ are obtained by performing \GBElimLin\ over some testset $T\subset \FF^k$,
for all $k'\leq k\leq k''$.

We finally notice that if $V_\FF(J)\neq\emptyset$ and we can estimate that the solution
is found on the average by means of a tamed (correct) $l$-guess with $l\leq k < k''$,
we can essentially consider $k$ as a final step for \MultiSolve\ even if $p_B(k)\neq 0$.
This generally implies a consistent reduction of the complexities $C_1,C_2$ in this
average case. We may have such an estimation, for instance, when studying polynomial
systems arising from the algebraic cryptanalysis of a cryptosystem by using a testset
of many different keys.

To show that the algorithm \MultiSolve\ can be practically used in the context
of cryptography, in the next sections we introduce and attack the well-known stream
cipher \Trivium\ \cite{DCP}. In particular, we present this cryptographic scheme
and its cryptanalysis in the general framework provided by the notion of ``difference
stream cipher'' that has been recently introduced in the papers \cite{LSPTV,LST}.

\section{Stream ciphers}

Stream ciphers are essentially practical realizations of pseudorandom functions
(see, for instance, \cite{AA,KaLi,MPS}). At a high level, on input of two fixed-length
random strings usually named {\em seed} and {\em initialisation vector}, a stream cipher
produces a random-looking string of arbitrary length which is called the {\em keystream}.
The characters in these strings are generally elements of a finite field and they are most
commonly bits. To obtain a symmetric cipher, the keystream can be used to encrypt a stream
of plaintexts or decrypt a stream of ciphertexts simply by addition or subtraction.
In this case, the seed and the initialization vector are respectively the key and the
nonce of the cipher. A stream cipher is deemed secure if it is indistinguishable
from a proper random function.

Most of the known constructions for stream ciphers rely on the use of feedback shift
registers - FSR, in short. An FSR is an array of memory cells, with each cell storing
a single element of a given finite field $\FF = \GF(q)$, which are updated by means of
some function $f$. More precisely, within an update the values of the cells are shifted
to the left while the right-most cell gets as a new value the image under the function $f$
of the current values of a given subset of cells of the same register and possibly that
of some other FSRs. For finite fields, the function $f$ can be always converted into
a multivariate polynomial and it is called therefore {\em update polynomial}.
An FSR is called {\em linear} or {\em non-linear} is the polynomial $f$ is linear
or not, respectively.

A stream cipher $\cC$ can be obtained by combining a few FSRs $\cF_1,\ldots,\cF_n$,
their corresponding update polynomials $f_1,\dots,f_n$ and a further multivariate
polynomial $g$ with coefficients in $\FF$ which is called {\em keystream polynomial}.
Indeed, the keystream of $\cC$ is obtained by evaluating the polynomial $g$ over the current
values of the registers of the $\cF_i$. In particular, the memory cells of the FSRs are
initially filled with the elements of a seed, an initialisation vector and, possibly,
a constant string. Then, the $\cF_i$ are updated in parallel $(h + u)$-times, where $h$
is the intended length of the keystream and $u$ is the number of updates (including initialization)
in the {\em warm-up stage}. This stage places some distance between the initialization
of the FSRs and the output of the keystream in order to prevent attacks on the stream
cipher that are based on the knowledge of some amount of elements in the keystream. 

A natural language for describing stream ciphers obtained from FSRs is that of explicit
difference equations, as first observed in \cite{LST} where such stream ciphers
are called {\em difference stream ciphers}. For the sake of readability, we recall
some notations before their formal definition (see also \cite{GLS}).

Given a finite field $\FF = \GF(q)$ and $n\in\NN^* = \NN\setminus\{0\}$, we denote
by $R$ the polynomial algebra $\FF[X]$ in the infinite set of variables
$X = \bigcup_{t\in\NN} \{x_1(t),\dots,x_n(t)\}$. The algebra endomorphism
$\sigma: R\rightarrow R$ defined by putting $\sigma(x_i(t)) = x_i(t+1)$ ($1\leq i\leq n, t\geq 0$)
is called the {\em shift map} of $R$. For any $r_1,\dots,r_n\in\NN$, the following
subset of variables 
\[
\bar{X} = \{x_1(0),\dots,x_1(r_1-1),\dots,x_n(0),\dots,x_n(r_n-1))\}
\]
defines a (finitely generated) subalgebra $\bR = \FF[\bX]\subset R$.

\begin{definition}
\label{DiffStream}
A {\em difference stream cipher} $\cC$ is a system of (algebraic ordinary)
explicit difference equations of the form
\begin{equation}
\label{dsys}
\left\{
\begin{array}{ccc}
x_1(r_1 + t) & = & \sigma^t(f_1), \\
& \vdots \\
x_n(r_n + t) & = & \sigma^t(f_n), \\
\end{array}
\right.
\quad (t\in\NN)
\end{equation}
together with a polynomial $g\in \bR$, where $r_1,\ldots,r_n\in\NN$ and $f_1,\ldots,f_n\in\bR$.
\end{definition}

An {\em $\FF$-solution of the system $(\ref{dsys})$} is a solution of all its explicit difference
equations, that is, an $n$-tuple $(a_1,\ldots,a_n)$ of maps from $\NN$ to $\FF$ such that,
given $v(t) = (a_1(t),\dots,a_1(r_1-1+t),\dots,a_n(t),\dots,a_n(r_n-1+t))$, it holds that
$a_i(r_i + t) = f_i(v(t))$, for every integer $t \geq 0$. We call $v(t)$ the {\em $t$-state} of
$(a_1,\dots,a_n)$ and in particular $v(0)$ is the {\em initial state}. The system (\ref{dsys})
has a unique $\FF$-solution once the initial state is fixed \cite[Thm.~2.4]{LST}.

\begin{remark}
A difference stream cipher $\cC$ defined by an explicit difference system $(\ref{dsys})$
and a polynomial $g \in \bR$ can be realized as a stream cipher obtained from a set
$\cF_1,\ldots,\cF_n$ of FSRs having update polynomials $f_1,\ldots,f_n$ and keystream
polynomial $g$. In particular, the variable $x_i(t)$ $(1\leq i\leq n)$ corresponds to the
symbolic value of the left-most cell of the register of $\cF_i$ at the clock $t\geq 0$.
If $(a_1,\ldots,a_n)$ is an $\FF$-solution of $(\ref{dsys})$, its initial state is composed
by a seed, an initialisation vector and, possibly, a constant string. If $v(t)\in\FF^r$
$(r = r_1 +  \ldots + r_n)$ is the $t$-state of $(a_1,\ldots,a_n)$, the function $b:\NN\to\FF$
such that $b(t) = g(v(t))$ for all $t\geq 0$, is called the {\em keystream of $(a_1,\ldots,a_n)$}.
By design, the $j$-th bit of the keystream can be set equal to $b(u+j-1)$, with $u$ being
the number of warm-up updates of the stream cipher. Vice versa, a stream cipher obtained
from some FSRs can be formally described as a difference stream cipher. 
\end{remark}

\section{Algebraic attacks}

As already mentioned, a secure stream cipher $\cC$ can be used to construct a CPA-secure symmetric
encryption scheme $\Pi = (\mathsf{Gen}, \mathsf{Enc}, \mathsf{Dec})$ where CPA stands for Chosen
Plaintext Attack (see \cite{KaLi}). In particular, the key-generation algorithm $\mathsf{Gen}$
takes in input a security parameter $\lambda \in \NN$ and returns a random seed $s$ as the key.
The encryption algorithm $\mathsf{Enc}$, on input of a seed $s$ and a message $m$, samples
a uniform initialisation vector $\mathrm{IV}$ and outputs the ciphertext $c = (\mathrm{IV},
m + \mathsf{str}(|m|))$, where $\mathsf{str}(|m|)$ denotes the keystream produced by $\cC$
on input of $s$ and $\mathrm{IV}$ which is truncated after $|m|$ elements (here $m$ is a string
of elements of the finite field $\FF$ and $|m|$ denotes its length). Decryption is then performed
by subtraction of the keystream $\mathsf{str}(|m|)$.

CPA-security is assessed by means of a security game where the adversary is granted access
to an encryption oracle. Therefore, the adversary has access to a polynomial number (in the
security parameter $\lambda$) of strings $(\mathrm{IV}, \mathsf{str}(|m|))$ with the seed $s$
fixed and $\mathrm{IV}$ and $m$ which vary. Such knowledge can then be used by the adversary
to pursue a key-recovery attack where he attempts to deduce information about the secret seed $s$. 

As a consequence, if $\cC$ is a difference stream cipher, the adversary has access to some
string $(b(u),\dots,b(u+h-1))$ where $h\in\NN$ is polynomial in $\lambda$ and $b$ is the keystream
of a given $\FF$-solution $(a_1,\dots,a_n)$ of the system (\ref{dsys}). This implies the knowledge
of the generators of the ideal
\[
J_{u,h} = \sum_{u \leq t < u + h} \langle \sigma^t(g) - b(t) \rangle \subset R.
\] 
We denote by $V_\FF(J_{u,h})$ the set of simultaneous $\FF$-solutions of the generators of $J_{u,h}$. 
Given the subset $S = \{ x_1(r_1) - f_1,\ldots,x_n(r_n) - f_n \} \subset R$ and the ideal
\[
I = \langle \sigma^t(f)\mid f \in S, t \geq 0 \rangle \subset R
\]
we consider $V_\FF(I + J_{u,h}) = V_\FF(I)\cap V_\FF(J_{u,h})$, where $V_\FF(I)$ is the set
of all $\FF$-solutions of the difference system $(\ref{dsys})$. Since the function $b$ is the keystream
of a given $\FF$-solution $(a_1,\dots,a_n)\in V_\FF(I)$, we have that
$(a_1,\ldots,a_n)\in V_\FF(I + J_{u,h})\neq\emptyset$. Indeed, for actual stream ciphers
one has that $\#(V_\FF(I + J_{u,h})) = 1$ for some sufficiently large $h\in \NN$. In other words,
there is a unique solution of $(\ref{dsys})$ that is compatible with a sufficiently long keystream.

Denote by $\bV_\FF(I + J_{u,h})\subset \FF^r$ the set of the initial states of the $\FF$-solutions
$(a_1,\ldots,a_n)\in V_\FF(I + J_{u,h})$. The following result is essentially a consequence
of Proposition \ref{elim}.

\begin{theorem}\cite[Thm.~5.4]{LST}
\label{keyeq}
Let $\bT:\bR\to\bR$ be the algebra endomorphism defined for all $1\leq i\leq n$ as follows
\[
\begin{array}{rcl}
x_i(0) & \mapsto & x_i(1), \\
         & \vdots  & \\
x_i(r_i-2) & \mapsto & x_i(r_i-1), \\
x_i(r_i-1) & \mapsto & f_i. \\
\end{array}
\]
Moreover, define the ideal
\[
J'_{u,h} = \sum_{u \leq t < u+h} \langle \bT^t(g) - b(t) \rangle\subset \bR.
\]
Then, it holds that $\bV_\FF(I + J_{u,h}) = V_\FF(J'_{u,h})$. 
\end{theorem}

Under the reasonable assumption that for a sufficiently large number $h$ of key\-stream elements
there is a unique $\FF$-solution $(a_1,\dots,a_n)$ compatible with them, any attack which aims
at determining the initial state of $(a_1,\dots,a_n)$, that is, $V_\FF(J'_{u,h})$ is said to be
an {\em algebraic (key-recovery) attack}. Note that the case in which we have spurious solutions,
that is, $\# V_\FF(J'_{u,h}) > 1$ can be identified by means of a \Gr\ basis computation, namely by
the linear dimension
\[
\# V_\FF(J'_{u,h}) = \dim_\FF \bR/(J'_{u,h} + L)
\]
where $L\subset\bR$ is the field equation ideal corresponding to the polynomial algebra $\bR$.
Recall in fact that $V_\FF(J'_{u,h}) = V(J'_{u,h} + L)$. To compute this solution set,
one can use \Gr\ bases according to Proposition \ref{uniqueGB} or any other kind of solver.
However, if the update polynomials $f_i$ are non-linear ones, the polynomial $\bT^t(g)$ may have
a high degree for large values of the clock $t$. Since $t\geq u$ in the definition of $J'_{u,h}$,
this occurs whenever the number $u$ of warm-up updates of the stream cipher $\cC$ is large,
as it is usually the case. High degrees surely affect the \Gr\ basis computation for $J'_{u,h} + L$
and hence a long warm-up stage is a good security strategy also with respect to algebraic attacks.

Note that the initial state usually contains the seed, the initial vector and some constant
elements as well, and therefore it would seem desirable to attack it since the actual unknown
entries are less than the total length of the registers. However, for the special class of invertible
difference stream ciphers, the initial state can be uniquely recovered by any internal state
and hence it is more convenient to attack the state at clock $u$ where the keystream starts
to output. Indeed, in this case we can assume $u = 0$ and compute $V_\FF(J'_{0,h})$ where
the generators of $J'_{0,h}$ have much lower degrees than those of the generators of $J'_{u,h}$.
Of course, in this case all the entries of the considered initial state are completely unknown.

\begin{definition}
We call an explicit difference system $(\ref{dsys})$ {\em invertible} if the algebra endomorphism
$\bT$ defined in Theorem \ref{keyeq} is actually an automorphism. A difference stream
cipher $\cC$ defined by an explicit difference system is said {\em invertible} accordingly.
\end{definition}

Note that the algebra endomorphism $\bT:\bR\to\bR$ has a corresponding polynomial map
$\T:\FF^r\to\FF^r$ ($r = r_1 +  \ldots + r_n$). If $v(t)\in\FF^r$ is the $t$-state
of a $\KK$-solution $(a_1,\ldots,a_n)$ of $(\ref{dsys})$ we have that $\T(v(t)) = v(t+1)$,
for all clocks $t\geq 0$. We call $\bT$ the {\em state transition endomorphism} and $\T$
the {\em state transition map of the explicit difference system $(\ref{dsys})$}. If $\bT$
is an automorphism, the map $\T$ is also invertible and we have that $\T^{-1}(v(t+1)) = v(t)$.
This implies that for invertible difference stream ciphers, we can uniquely obtain
the initial state $v(0)$ starting from any internal state $v(t) = \T^t(v(0))$.

To establish invertibility and compute the inverse of an algebra automorphism one has
the following general result based on \Gr\ bases theory. For a comprehensive
reference about automorphisms of polynomial algebras we refer to the book \cite{VDE}.

\begin{theorem}
\label{invth}
Let $\KK$ be any field and let $X = \{x_1,\ldots,x_r\}, X' = \{x'_1,\ldots,x'_r\}$ be two
disjoint variable sets. Define the polynomial algebras $P = \KK[X], P' = \KK[X']$
and $Q = \KK[X\cup X']$. Consider an algebra endomorphism $\varphi:P\to P$ such that
$x_1\mapsto g_1,\ldots,x_r\mapsto g_r$ $(g_i\in P)$ and the corresponding ideal
$J\subset Q$ which is generated by the set $\{x'_1 - g_1,\ldots, x'_r - g_r\}$.
Moreover, endow the polynomial algebra $Q$ by a product monomial ordering such that
$X\succ X'$. Then, the map $\varphi$ is an automorphism of $P$ if and only if
the reduced \Gr\ basis of $J$ is of the kind $\{x_1 - g'_1,\ldots,x_r - g'_r\}$
where $g'_i\in P'$, for all $1\leq i\leq r$. In this case, if $\varphi':P'\to P'$
is the algebra endomorphism such that $x'_1\mapsto g'_1, \ldots, x'_r\mapsto g'_r$
and $\xi: P\to P'$ is the isomorphism $x_1\mapsto x'_1, \ldots, x_r\mapsto x'_r$,
we have that $\xi\, \varphi^{-1} = \varphi'\, \xi$.
\end{theorem}

The invertibility property is satisfied, for instance, by the stream cipher \Trivium\
on which we will experiment our algorithm \MultiSolve\ by performing an algebraic attack.
Note that the notion of inverse system of an invertible explicit difference system can be
introduced in a natural way \cite[Def.~3.6]{LST}. Moreover, invertible systems (having
independent subsystems) can be used to define block ciphers \cite[Def.~5.6]{LST}.

\section{Trivium}

The stream cipher \Trivium\ was proposed by De Canni\`ere and Preneel in 2005. It was submitted
to the eSTREAM competition and therein selected as part of the final portfolio. We recall that
eSTREAM was an European project which aimed at identifying new secure and efficient stream ciphers.
Despite a wide cryptanalytic effort, best known attacks against \Trivium\ are still significantly
slower than brute-force attacks. \Trivium\ is a difference stream cipher whose corresponding explicit
difference system consists only of three quadratic equations over the base field $\FF = \GF(2)$,
namely
\begin{equation}
\label{tri}
\left\{
\begin{array}{r@{\ }c@{\ }l}
x(93+t) & = & z(t) + x(24+t) + z(45+t) + z(1+t)z(2+t), \\
y(84+t) & = & x(t) + y(6+t) + x(27+t) + x(1+t)x(2+t), \\
z(111+t) & = & y(t) + y(15+t) + z(24+t) + y(1+t)y(2+t). \\
\end{array}
\right.
\quad (t\in\NN)
\end{equation}
Its keystream polynomial is the following homogeneous linear polynomial
\begin{equation}
\label{trikpol}
g  =  x(0) + x(27) + y(0) + y(15) + z(0) + z(45).
\end{equation}
Consequently, a $t$-state is a string of $288 = 93 + 84 + 111$ bits, for any $t\geq 0$.
The number of warm-up updates is $u = 4\cdot 288 = 1152$. Seeds (\ie\ private keys) and
initialisation vectors are 80-bit strings and together they form 160 bits of an initial state.
The remaining 128 bits are constants.

If $\bR = \FF[x(0),\ldots,x(92),y(0),\ldots,y(83),z(0),\ldots,z(110)]$, the state transition
endomorphism $\bT:\bR\to\bR$ of \Trivium\ is therefore the following one
\begin{equation}
\label{triauto}
\small
\begin{array}{l}
x(0)\mapsto x(1),\ldots,x(91)\mapsto x(92), x(92)\mapsto z(0) + x(24) + z(45) + z(1)z(2), \\
y(0)\mapsto y(1),\ldots,y(82)\mapsto y(83), y(83)\mapsto x(0) + y(6) + x(27) + x(1)x(2), \\
z(0)\mapsto z(1),\ldots,z(109)\mapsto z(110), z(110)\mapsto y(0) + y(15) + z(24) + y(1)y(2). \\
\end{array}
\end{equation}
By means of Theorem \ref{invth} we obtain that $\bT$ is in fact an automorphism whose
inverse $\bT^{-1}:\bR\to\bR$ is defined as
\begin{equation}
\small
\begin{array}{l}
x(92)\mapsto x(91),\ldots,x(1)\mapsto x(0), x(0)\mapsto y(5) + x(26) + y(83) + x(0)x(1), \\
y(83)\mapsto y(82),\ldots,y(1)\mapsto y(0), y(0)\mapsto y(14) + z(23) + z(110) + y(0)y(1), \\
z(110)\mapsto z(109),\ldots,z(1)\mapsto z(0), z(0)\mapsto x(23) + z(44) + x(92) + z(0)z(1). \\
\end{array}
\end{equation}
The invertibility of the stream cipher \Trivium\ allows therefore an algebraic attack on
its $u$-state that we have analyzed by means of the algorithm \MultiSolve. In fact,
the high number of variables, namely 288, makes impossible to solve the corresponding polynomial
system by a single \Gr\ basis computation. Our multistep strategy divides instead
this solving task in many subproblems, whose number can be easily estimated, where the number
of variables is bounded at will.

\section{A multistep attack on Trivium}

In this section we essentially introduce the reader to the practical use of the
complexity analysis of \MultiSolve\ contained in Section 4. We make use to this
purpose of an algebraic attack on the stream cipher \Trivium\ whose computational
cost we are able to estimate. This cryptanalysis is an instance of an algebraic attack
on an internal state of an invertible difference stream cipher as described in Section 6.
Recall that for \Trivium\ we consider the polynomial algebra
$\bR = \FF[x(0),\ldots,x(92),y(0),\ldots,y(83),z(0),\ldots,z(110)]$ ($\FF = \GF(2)$)
and the corresponding field equation ideal $L\subset \bR$.
We apply \MultiSolve\ to the generating set $H$ of the ideal $J'_{0,h} + L\subset \bR$
where
\[
J'_{0,h} = \sum_{0 \leq t < h} \langle \bT^t(g) - b(t) \rangle\subset \bR.
\]
Here $\bT:\bR\to\bR$ is the state transition automorphism (\ref{triauto}) of \Trivium,
$g\in\bR$ is its keystream polynomial (\ref{trikpol}) and $b(t)$ ($0\leq t < h$)
is the portion of the keystream that we assume to know for the attack. The goal of the attack
is to compute $V_\FF(J'_{0,h}) = V(J'_{0,h} + L)$, that is, the internal state of \Trivium\
at the keystream initial clock (see Section 6).

In our attack, the length of the keystream is set to $h = 240$ because this is a minimal
amount of data providing at most one solution in all polynomial systems we computed
to analyze \MultiSolve, as well as polynomials of relatively low degree ($\leq 5$) in the
generating set $H$.

A first useful observation is that the generators $\bT^t(g) - b(t)\in H$ ($0\leq t\leq 65$)
are independent linear polynomials and hence we can eliminate immediately 66 variables
by means of them. The remaining set of 222 variables we consider is
\[
\bX' = \{x(0),\ldots,x(92),y(0),\ldots,y(83),z(0),\ldots,z(44)\}.
\]
This set is still too large to actually perform \GBElimLin\ and hence we decide to fix
a set $V$ of 116 variables to be stepwise evaluated. To define the set $V$ we start observing that,
owing to the defining equations of \Trivium, there are three subsets of $\bX'$, each consisting
of 73 variables, whose evaluation results in having many linear polynomials
within the system that needs to be solved. These subsets are the following ones:
\begin{equation*}
\footnotesize
\begin{gathered}
V_0 = \{x(0),y(0),z(0),x(3),y(3),z(3),x(6),y(6),z(6),x(9),y(9),z(9),x(12),y(12),z(12), \\
x(15),y(15),z(15),x(18),y(18),z(18),x(21),y(21),z(21),x(24),y(24),z(24),x(27),y(27),z(27), \\
x(30),y(30),z(30),x(33),y(33),z(33),x(36),y(36),z(36),x(39),y(39),z(39),x(42),y(42),z(42), \\
x(45),y(45),x(48),y(48),x(51),y(51),x(54),y(54),x(57),y(57),x(60),y(60),x(63),y(63),x(66), \\
y(66),x(69),y(69),x(72),y(72),x(75),y(75),x(78),y(78),x(81),x(84),x(87),x(90)\}; \\
\end{gathered}
\end{equation*}
\begin{equation*}
\footnotesize
\begin{gathered}
V_1 = \{x(1),y(1),z(1),x(4),y(4),z(4),x(7),y(7),z(7),x(10),y(10),z(10),x(13),y(13),z(13), \\
x(16),y(16),z(16),x(19),y(19),z(19),x(22),y(22),z(22),x(25),y(25),z(25),x(28),y(28),z(28), \\
x(31),y(31),z(31),x(34),y(34),z(34),x(37),y(37),z(37),x(40),y(40),z(40),x(43),y(43),z(43), \\
x(46),y(46),x(49),y(49),x(52),y(52),x(55),y(55),x(58),y(58),x(61),y(61),x(64),y(64),x(67), \\
y(67),x(70),y(70),x(73),y(73),x(76),y(76),x(79),y(79),x(82),x(85),x(88),x(91)\}; \\
\end{gathered}
\end{equation*}
\begin{equation*}
\footnotesize
\begin{gathered}
V_2 = \{x(2),y(2),z(2),x(5),y(5),z(5),x(8),y(8),z(8),x(11),y(11),z(11),x(14),y(14),z(14), \\
x(17),y(17),z(17),x(20),y(20),z(20),x(23),y(23),z(23),x(26),y(26),z(26),x(29),y(29),z(29), \\
x(32),y(32),z(32),x(35),y(35),z(35),x(38),y(38),z(38),x(41),y(41),z(41),x(44),y(44),z(44), \\
x(47),y(47),x(50),y(50),x(53),y(53),x(56),y(56),x(59),y(59),x(62),y(62),x(65),y(65),x(68), \\
y(68),x(71),y(71),x(74),y(74),x(77),y(77),x(80),y(80),x(83),x(86),x(89),x(92)\}; \\
\end{gathered}
\end{equation*}
The set $V$ is thus chosen by combining one of the above sets with a subset of another,
consisting of 43 variables. Using \GBElimLin, we search for a set $V$ having the lowest average
value of $\NRV$, for a testset of different keys and evaluations. Some amount of experiments
suggests to define $V = V_2\cup V'_0$ where
\begin{equation*}
\footnotesize
\begin{gathered}
V'_0 = \{x(3),y(3),z(3),x(6),y(6),z(6),x(9),y(9),z(9),x(12),y(12),z(12),x(15),y(15),z(15), \\
x(18),y(18),z(18),x(21),y(21),z(21),x(24),y(24),z(24),x(27),y(27),z(27),x(30),y(30),z(30), \\
x(33),y(33),z(33),x(36),y(36),z(36),x(39),y(39),z(39),x(42),y(42),z(42),y(45)\}. \\
\end{gathered}
\end{equation*}

For our multistep strategy, we begin by evaluating 106 variables obtained by removing
the last 10 variables from the set $V'_0$ and we proceed by evaluating one additional
variable at each step. We label these steps with the corresponding number of evaluated
variables, say $k$, and thus we fix $106 \leq k\leq 116$ in our experiments. We apply
\GBElimLin\ with parameter $D$ set to the maximum input degree which is generally equal
to 3. This yields a quite low average running time for \GBElimLin\ which is, for instance,
4.5 seconds for $k = 106$ in our experiments.

For the testing activities, we run our code on a server with the following hardware
configuration: 
\begin{itemize}
\item CPU: 2 x AMD EPYC(TM) 7742;
\item Number of CPU Cores/Threads: 2 x 64 Cores/2 x 128 Threads;
\item Maximum CPU frequency achievable: 3.4GHz;
\item L3 cache: 256Mb;
\item RAM: 2 TB.
\end{itemize}
On this server, we install a Linux distribution as operating system -- Ubuntu 22.04 LTS
-- and \Magma\ version 2.27, a software package designed to solve problems in algebra
and number theory \cite{Magma}. We set \Magma\ to use 16 threads when performing parallel
linear algebra. Our main testset consists of $2^{16}$ different tests, namely $2^{12}$
random guesses of variables for $2^4$ keystreams obtained starting from a same number
of random initial states. We call this set a {\em random testset}. We also use a testset
of $2^{16}$ different correct guesses corresponding to $2^{16}$ random initial states
and we call it a {\em correct testset}. Recall that an initial state is here the internal
state of \Trivium\ at the initial clock of the keystream.

The complexity of the algorithm \MultiSolve\ strictly depends on the parameter $B$
which defines the tamed guesses, that is, the \Gr\ bases that are actually computed.
In our experiments, we fix $32\leq B\leq 38$. One reason for this choice is that
we found $p_{32}(116) = 0$, that is, $k'' = 116$ is the possible last step for $B = 32$
where $116$ is the maximum step considered. Moreover, we have that $\NRV\leq B = 38$
allows the computation of \GB\ by \Magma\ in few minutes at most which is for us
``a reasonable time'' for quite large testsets. Indeed, the solving degree
we generally have in our tests is at most 5. For a number of variables greater
than 38, we also observed that some internal errors occur when performing \GB\
with the parameter ``HFE''. Indeed, we adopted this option in all our computations
as it generally speeds them up. This is mainly because, with this parameter set,
\Magma\ employs the F4 algorithm with parallel computations on dense Macaulay matrices.
For this reason as well, namely to ensure completely successful and fast computations,
we chose to set $B\leq 38$.

The main statistics for computing the complexity of \MultiSolve\ are the probabilities
$p_B(k)$ for $B$ and $k$ in the considered intervals. To this purpose, we actually ran
two independent random testsets on our server finding no significative differences
between the two statistics obtained. Recall that these testsets consist each
of $2^{12}$ random guesses for $2^4$ random initial states.

\suppressfloats[b]
\begin{table}[ht!]
\caption{Probabilities $p_B(k)$ estimated on our {\em random} testsets.}
\begin{tabular}{|c|l|l|l|l|l|l|l|}
\hline
$k / B$
& \multicolumn{1}{|c|}{32} & \multicolumn{1}{|c|}{33} & \multicolumn{1}{|c|}{34}
& \multicolumn{1}{|c|}{35} & \multicolumn{1}{|c|}{36} & \multicolumn{1}{|c|}{37}
& \multicolumn{1}{|c|}{38} \\
\hline
106 & 0.63153 & 0.61703 & 0.60867 & 0.58736 & 0.55925 &	0.50848	& 0.44923 \\
\hline
107 & 0.61670 & 0.60675 & 0.58188 & 0.55359 & 0.50471 &	0.44638	& 0.38121 \\
\hline
108 & 0.57581 & 0.54169 & 0.49049 & 0.43468 & 0.37077 &	0.31258	& 0.23341 \\
\hline
109 & 0.49127 & 0.43814 & 0.38190 & 0.29349 & 0.23219 &	0.16003	& 0.10303 \\
\hline
110 & 0.43263 & 0.37582 & 0.28784 & 0.22722 & 0.15845 &	0.10179	& 0.04910 \\
\hline
111 & 0.28415 & 0.22218 & 0.14967 & 0.09698 & 0.04762 &	0.02165	& 0.00670 \\
\hline
112 & 0.14362 & 0.08954 & 0.04715 & 0.01627 & 0.00650 &	0.00053	& 0	  \\
\hline
113 & 0.08838 & 0.04610 & 0.01549 & 0.00650 & 0.00053 &	0	& 0	  \\
\hline
114 & 0.01498 & 0.00725 & 0.00053 & 0	    & 0	      &	0	& 0	  \\
\hline
115 & 0.00043 & 0	& 0	  & 0	    & 0	      &	0	& 0	  \\
\hline
116 & 0	      & 0	& 0	  & 0	    & 0	      &	0	& 0	  \\
\hline
\end{tabular}
\end{table}

The Table 1 provide us with the possible final steps $k''$ for each value
of $B$, in the worst case where the correct guess corresponding to the initial state
becomes tamed at such steps.

\suppressfloats[b]
\begin{table}[ht!]
\caption{Final steps $k''$ estimated on our {\em random} testsets (worst case).}
\begin{tabular}{|c|c|c|c|c|c|c|c|}
\hline
B      & 32  & 33  & 34  & 35  & 36  & 37  & 38 \\
\hline
$k''$ & 116 & 115 & 115 & 114 & 114 & 113 & 112 \\
\hline
\end{tabular}
\end{table}

Note that the above final steps gives us essentially the complexity $C = 2^{k''}$
of the one-step strategy, that is, the standard guess-and-determine or hybrid strategy. 
Assuming that the first step $k'$ of the full multistep strategy is set $k' = 106$
for all $B$, we can compare $C$ with main complexity of the algorithm
\MultiSolve, namely
\[
C_2 = \sum_{k'\leq k\leq k''} (p_B(k-1) - p_B(k)) 2^k.
\]
Recall that by convention $p_B(k'-1) = 1$ and Theorem \ref{triumph} implies that
$C > C_2$. The Table 3 and Figure 1 provide a comparison between the one-step
and the multistep strategy for our algebraic attack on \Trivium.

\suppressfloats[b]
\begin{table}[ht!]
\caption{One-step complexity VS Multistep complexity (worst case).}
\begin{tabular}{|c|c|c|c|c|c|c|c|}
\hline
B             & 32  & 33  & 34  & 35  & 36  & 37  & 38 \\
\hline
$\log_2(C)$ & 116 & 115 & 115 & 114 & 114 & 113 & 112 \\
\hline
$\log_2(C_2)$ &
111.63 &
111.13 &
110.47 &
109.93 &
109.37 &
108.85 &
108.29 \\
\hline
\end{tabular}
\end{table}

\suppressfloats[b]
\begin{figure}
\caption{One-step complexity VS Multistep complexity (worst case)}
\centering
\begin{tikzpicture} 
\begin{axis}[ 
        xlabel = B,
        ylabel = {$\log_2(\mathrm{complexity})$},
        ybar, 
        ymin=106,ymax=117, 
        enlarge x limits=0.1, 
        symbolic x coords={32,33,34,35,36,37,38}, 
        xtick=data,
        legend cell align=left,
        legend pos = north east,
        nodes near coords align={vertical}, 
        width=10cm, 
        height=7cm,  
        ymajorgrids=true,
        every node near coord/.append style={
            anchor=north, yshift=-0.07ex, font=\tiny
        },
        tick label style={font=\footnotesize},
        xticklabel style={yshift=-0.25ex},
        grid style=dashed,
] 

\addplot[fill=gray!95] 
coordinates {(32,116) (33,115) (34,115) (35,114) (36,114) (37,113) (38,112)}; 

\addplot[fill=gray!25]
coordinates {
(32,111.63)
(33,111.13)
(34,110.47)
(35,109.93)
(36,109.37)
(37,108.85)
(38,108.29)
}; 
\legend{One-step complexity, Multistep complexity (worst case)} 
\end{axis} 
\end{tikzpicture} 
\end{figure}

To obtain the complexity of \MultiSolve\ in the average case, we need a second
statistics, namely the probabilities $p_B(k)$ in the case of the correct testset.
In order to prevent confusion, let us denote by $\bar{p}_B(k)$ the probabilities
obtained in this way. For the following statistics, we make use of two independent
testsets of $2^{16}$ different initial states each.

\suppressfloats[b]
\begin{table}[ht!]
\caption{Probabilities $\bar{p}_B(k)$ estimated on our {\em correct} testsets.}
\begin{tabular}{|c|l|l|l|l|l|l|l|}
\hline
$k / B$
& \multicolumn{1}{|c|}{32} & \multicolumn{1}{|c|}{33} & \multicolumn{1}{|c|}{34}
& \multicolumn{1}{|c|}{35} & \multicolumn{1}{|c|}{36} & \multicolumn{1}{|c|}{37}
& \multicolumn{1}{|c|}{38} \\
\hline
106 & 0.98145 & 0.96414 & 0.93700 & 0.89612 & 0.83978 & 0.76630 & 0.67615 \\
\hline
107 & 0.96077 & 0.93288 & 0.89172 & 0.83521 & 0.76161 & 0.67128 & 0.56934 \\
\hline
108 & 0.88625 & 0.82880 & 0.75348 & 0.66212 & 0.55914 & 0.44754 & 0.33824 \\
\hline
109 & 0.75659 & 0.66649 & 0.56439 & 0.45430 & 0.34366 & 0.24200 & 0.15723 \\
\hline
110 & 0.66171 & 0.56049 & 0.45132 & 0.34123 & 0.24028 & 0.15625 & 0.08948 \\
\hline
111 & 0.44080 & 0.33250 & 0.23412 & 0.15137 & 0.08672 & 0.04370 & 0.01810 \\
\hline
112 & 0.23505 & 0.15071 & 0.08466 & 0.04102 & 0.01611 & 0.00458 & 0.00078 \\
\hline
113 & 0.14989 & 0.08432 & 0.04094 & 0.01608 & 0.00458 & 0.00078 & 0       \\
\hline
114 & 0.03961 & 0.01570 & 0.00450 & 0.00075 & 0       & 0       & 0       \\
\hline
115 & 0.00381 & 0.00056 & 0       & 0       & 0       & 0       & 0       \\
\hline
116 & 0.00056 & 0       & 0       & 0       & 0       & 0       & 0       \\
\hline
\end{tabular}
\end{table}

We use the above statistics to determine at which step $\bar{k}''$ one has
probability $\bar{p}_B(\bar{k}'') < 0.5$, that is, at least half of the correct
guesses become tamed in a step $l\leq \bar{k}'' < k''$. In other words,
this provides that with probability greater than or equal to $0.5$,
the algorithm \MultiSolve\ will stop at a final step $l\leq \bar{k}''$.

\suppressfloats[b]
\begin{table}[ht!]
\caption{Final steps $\bar{k}''$ estimated on our {\em correct} testsets (average case).}
\begin{tabular}{|c|c|c|c|c|c|c|c|}
\hline
B                 & 32  & 33  & 34  & 35  & 36  & 37  & 38 \\
\hline
$\vphantom{\bar{\bar{M}}} \bar{k}''$ & 111 & 111 & 110 & 109 & 109 & 108 & 108 \\
\hline
\end{tabular}
\end{table}

By considering $k''$ and $\bar{k}''$ as last steps for the algorithm \MultiSolve,
the following Table 6 and Figure 2 provide a comparison between the worst case
complexity $C_2$ and the average case complexity
\[
\bar{C}_2 = \sum_{k'\leq k\leq \bar{k}''} (p_B(k-1) - p_B(k)) 2^k.
\]

\suppressfloats[b]
\begin{table}[ht!]
\caption{Multistep complexity (worst case) VS Multistep complexity (average case).}
\begin{tabular}{|c|c|c|c|c|c|c|c|}
\hline
B                   & 32  & 33  & 34  & 35  & 36  & 37  & 38 \\
\hline
$\log_2(C_2)$ &
111.63 &
111.13 &
110.47 &
109.93 &
109.37 &
108.85 &
108.29 \\
\hline
$\vphantom{\bar{\bar{M}}} \log_2(\bar{C}_2)$ &
108.79 &
108.88 &
107.67 &
107.06 &
107.13 &
106.20 &
106.35 \\
\hline
\end{tabular}
\end{table}

\suppressfloats[b]
\begin{figure}
\caption{Multistep complexity (worst case) VS Multistep complexity (average case)}
\centering
\begin{tikzpicture} 
\begin{axis}[ 
        xlabel = B,
        ylabel = {$\log_2(\mathrm{complexity})$},
        ybar, 
        ymin=106,ymax=112, 
        enlarge x limits=0.1, 
        symbolic x coords={32,33,34,35,36,37,38}, 
        xtick=data,
        legend cell align=left,
        legend pos = north east,
        nodes near coords align={vertical}, 
        width=10cm, 
        height=7cm,  
        ymajorgrids=true,
        every node near coord/.append style={
            anchor=north, yshift=-0.07ex, font=\tiny
        },
        tick label style={font=\footnotesize},
        xticklabel style={yshift=-0.25ex},
        grid style=dashed,
] 

\addplot[fill=gray!95] 
coordinates {
(32,111.63)
(33,111.13)
(34,110.47)
(35,109.93)
(36,109.37)
(37,108.85)
(38,108.29)
}; 

\addplot[fill=gray!25]
coordinates {
(32,108.79)
(33,108.88)
(34,107.67)
(35,107.06)
(36,107.13)
(37,106.20)
(38,106.35)
}; 
\legend{Multistep complexity (worst case), Multistep complexity (average case)} 
\end{axis} 
\end{tikzpicture} 
\end{figure}

We emphasize that the logarithm of our best complexity, namely $\log_2(\bar{C}_2) = 106.2$
for $B = 37$, is very close to the minimum number $k' = 106$ of evaluated variables
that we have chosen for our optimization. This occurs with $p_{37}(106) = 0.5$, that is,
with half of the \Gr\ bases in the step 106 that are not computed because $\NRV > 37$.
Using our code in \Magma\ running on our server, for this average case complexity
we estimate that an algebraic attack on \Trivium\ by \MultiSolve\ takes approximately
$2^{112}$ seconds. This estimation is obtained as $T = T_1 + T_2$ (see formulas
(\ref{runtime})) where the running time of the \Gr\ bases only is $T_2 = 2^{111.5}$
seconds. This clearly shows that the running time $T_1$ of the calls at \GBElimLin\
is irrelevant to the total.


\section{Conclusions and further directions}

In this paper we have shown that the multistep strategy is, in theory and
in practice, a viable way to estimate and improve the complexity of solving
polynomial systems over finite fields. Although the neat stream cipher \Trivium\
remains secure after our algebraic attack based on this solving strategy,
we obtain a relevant reduction to $2^{106.2}$ of the required complexity with
respect to previous attacks of the same kind.

Indeed, the experimentation of the proposed method is only at the beginning
and we believe that a further optimization of the multiple parameters of the
multistep strategy, along with continuous progress in solving algorithms,
will make this technique even more effective as a tool to assess the security
level of a cryptographic system. Of course, the multistep strategy is a general
algorithm that could also be useful outside the cryptographic context.

As a further direction in this line of research, we suggest that the limitation
of the number of remaining variables can be generalized to bounds involving
other critical indicators that a polynomial system could be difficult to solve,
such as the degree of its equations, their density, a theoretical value of
the solving degree and so on. Of course, these choices heavily depend on the kind
of polynomial system solver one wants to use since the multistep strategy
can employ any of them.

To conclude, we believe that the concept of statistical evaluation of the
complexity of a polynomial system over a finite field is promising and requires
further investigation.

\section{Acknowledgements}

First, we would like to thank Lorenzo D'Ambrosio for fruitful discussions
and suggestions. Moreover, we gratefully acknowledge the use of the HPC resources
of {\sc ReCaS-Bari} (University of Bari, INFN).

We would like to express our gratitude to the anonymous referees for their thorough
review of the manuscript. We appreciate all the insightful comments and suggestions,
as they have significantly enhanced the overall readability of the paper.


\begin{thebibliography}{00}

\bibitem{AA} Armknecht, F., Ars, G., Algebraic Attacks on Stream Ciphers with \Gr\ Bases.
\Gr\ Bases, Coding and Cryptography. Springer, Berlin, 2009.

\bibitem{AL} Adams, W.W.; Loustaunau, P., An introduction to Gr\"obner bases.
Graduate Studies in Mathematics, 3. American Mathematical Society, Providence, RI, 1994.

\bibitem{APS} Amadori, A.; Pintore, F.; Sala, M., On the discrete logarithm problem for prime-field
elliptic curves. Finite Fields Appl., 51 (2018), 168--182.

\bibitem{Ba} Bard, G.V., Algebraic cryptanalysis. Springer, Dordrecht, 2009.

\bibitem{BFSS} Bardet, M.; Faug\`ere, J.-C.; Salvy, B.; Spaenlehauer, P.-J., On the complexity
of solving quadratic Boolean systems. J. Complexity 29 (2013), no. 1, 53--75.

\bibitem{BFSY} Bardet, M.; Faug\`ere, J.-C.; Salvy, B.; Yang, B.-Y., Asymptotic behaviour of the degree
of regularity of semi-regular polynomial systems. Effective Methods in Algebraic Geometry -
MEGA 2005, 1--14, 2005.

\bibitem{BFP1} Bettale, L.; Faug\`ere, J.-C.; Perret, L., Hybrid approach for solving
multivariate systems over finite fields. J. Math. Cryptol. 3 (2009), no. 3, 177--197.

\bibitem{BFP2} Bettale, L.; Faug\`ere, J.-C., Perret, L., Solving Polynomial Systems over Finite Fields:
Improved Analysis of the Hybrid Approach. International Symposium on Symbolic and Algebraic
Computation - ISSAC 2012, 67--74, 2012.

\bibitem{BorKnuMat} Borghoff, J.; Knudsen, L.R.; Matusiewicz, K., Hill Climbing Algorithms
and Trivium. Selected Areas in Cryptography. SAC 2010, Lecture Notes in Comput. Sci.,
6544. Springer, Berlin, 2011.

\bibitem{Magma} Bosma, W.; Cannon, J.J.; Fieker, C.; Steel, A. (eds.), Handbook of Magma
functions, Edition 2.27 (2022), 6359 pages. 

\bibitem{CaGo} Caminata, A.; Gorla E., Solving Multivariate Polynomial Systems
and an Invariant from Commutative Algebra. Arithmetic of Finite Fields, 
Lecture Notes in Comput. Sci., 12542, 3--36, Springer, Cham, 2021.

\bibitem{Ped} Cianfriglia, M.; Onofri, E.; Onofri, S.; Pedicini, M., Fourteen years of cube attacks,
Appl. Algebra Engrg. Comm. Comput., 2023, 41 pp.

\bibitem{CB} Courtois, N.; Bard, G.V., Algebraic cryptanalysis of the data encryption standard.
Cryptography and Coding, Lectures Notes in Comput. Sci., 4887, 152--169, Springer, Berlin, 2007.

\bibitem{CKPS} Courtois, N.; Klimov, A.; Patarin, J.; Shamir, A., Efficient algorithms for solving
overdefined systems of multivariate polynomial equations. Advances in Cryptology - EUROCRYPT 2000,
Lecture Notes in Comput. Sci., 1807, 392--407, Springer, Berlin, 2000.

\bibitem{DCP} De Canni\`ere, C.; Preneel, B., Trivium Specifications. eSTREAM -
ECRYPT Stream Cipher Project, \newblock {https://www.ecrypt.eu.org/stream/triviumpf.html}.

\bibitem{Cube} Dinur, I.; Shamir, A., Cube Attacks on Tweakable Black Box Polynomials.
Advances in cryptology - EUROCRYPT 2009, Lecture Notes in Comput. Sci., 5479, 278--299,
Springer, Berlin, 2009.

\bibitem{Du} Dub\`e, T.W., The structure of polynomials ideals and \Gr\ bases,
SIAM J. Comput., 19 (1990), 750--773.

\bibitem{F4} Faug\`ere, J.C., A new efficient algorithm for computing Gr\"obner bases (F4).
Effective Methods in Algebraic Geometry - MEGA 1998. J. Pure Appl. Algebra, 139 (1999),
no. 1-3, 61--88.

\bibitem{F5} Faug\`ere, J.C., A new efficient algorithm for computing \Gr\ bases without
reduction to zero (F5). International Symposium on Symbolic and Algebraic Computation - ISSAC 2002,
75--83, ACM, New York, 2002.

\bibitem{FGLM} Faug\`ere, J. C.; Gianni, P.; Lazard, D.; Mora, T., Efficient computation
of zero-dimensional Gr\"obner bases by change of ordering. J. Symbolic Comput. 16 (1993),
329--344. 

\bibitem{FJ} Faug\`ere, J.-C.; Joux, A., Algebraic cryptanalysis of hidden field equation (HFE)
cryptosystems using \Gr\ bases. Advances in Cryptology - CRYPTO 2003, Lecture Notes in Comput.
Sci., 2729, 44--60, Springer, Belin, 2003.

\bibitem{GaoHua}
Gao, X.-S.; Huang, Z., Characteristic set algorithms for equation solving in finite fields.
J. Symbolic Comp. 47 (2012), no. 6, 655--679.

\bibitem{NPMP} Garey, M.-R.; Johnson, D.-S., Computers and intractability: a guide to the
theory of NP-completeness. Freeman, 1979.

\bibitem{GB} Gerdt, V.P.; Blinkov, Y.A., Involutive bases of polynomial ideals. Simplification
of systems of algebraic and differential equations with applications. Math. Comput.
Simulation 45 (1998), no. 5-6, 519--541.

\bibitem{GLS} Gerdt, V.; La Scala, R., Noetherian quotients of the algebra
of partial difference polynomials and Gr\"obner bases of symmetric ideals.
J. Algebra 423 (2015), 1233--1261.

\bibitem{Gh} Ghorpade, S.R., A note on Nullstellensatz over finite fields.
Contributions in Algebra and Algebraic Geometry, 23--32, Contemp. Math.,
738, Amer. Math. Soc., Providence, RI, 2019.

\bibitem{HL} Hashemi, A.; Lazard, D., Sharper complexity bounds for zero-dimensional
\Gr\ bases and polynomial system solving, Internat. J. Algebra Comput. 21 (2011), 703--713.

\bibitem{HeHuPrWa}
He, J.; Hu, K.; Preneel, B.; Wang, M., Stretching Cube Attacks: Improved Methods to Recover
Massive Superpolies. Advances in Cryptology - ASIACRYPT 2022. Lecture Notes in Comput. Sci.,
13794. Springer, Cham, 2022. 

\bibitem{HL} Huang, Z.; Lin, D., Attacking Bivium and Trivium with
the characteristic set method. Progress in cryptology - AFRICACRYPT 2011,
Lecture Notes in Comput. Sci., 6737, 77--91, Springer, Berlin, 2011.

\bibitem{KaLi} Katz, J.; Lindell, Y., Introduction to modern cryptography. Second edition. 
Chapman \& Hall/CRC Cryptography and Network Security. CRC press, Boca Raton, FL, 2015.

\bibitem{LSPTV} La Scala, R.; Polese S.; Tiwari S.K.; Visconti A., An algebraic attack
to the Bluetooth stream cipher E0. Finite Fields Appl. 84 (2022), Paper No. 102102, 29 pp.

\bibitem{LST} La Scala, R.; Tiwari S.K., Stream/block ciphers, difference equations and
algebraic attacks. J. Symbolic Comput. 109 (2022), 177--198.

\bibitem{Ma} Marek, V.W., Introduction to mathematics of satisfiability, Chapman \& Hall/CRC Studies
in Informatics Series. CRC Press, Boca Raton, FL, 2009. 

\bibitem{MPS} Mascia, C.; Piccione, E.; Sala, M., An algebraic attack on stream ciphers with
application to nonlinear filter generators and WG-PRNG. arXiv:2112.12268, 2022, 16 pp.,
to appear in Adv. Math. Commun.

\bibitem{MaxBir} Maximov, A.; Biryukov, A., Two Trivial Attacks on Trivium.
Selected Areas in Cryptography - SAC 2007. Lecture Notes in Comput. Sci., 4876. Springer,
Berlin, 2007. 

\bibitem{OS} Orsini, E.; Sala, M., Correcting errors and erasures via the syndrome variety.
J. Pure Appl. Algebra 200 (2005), no. 1-2, 191--226.

\bibitem{Rad} Raddum, H., Cryptanalytic Results on Trivium. eSTREAM - ECRYPT Stream Cipher
Project, Report 2006/039 (2006).

\bibitem{RM} Rajchel-Mieldzio\'c, G., Quantum mappings and designs, PhD thesis,
Centrum Fizyki Teoretycznej Polskiej Akademii Nauk, 2022.

\bibitem{TII} Ramos-Calderer, S.; Bravo-Prieto, C.; Lin, R.; Bellini, E.; Manzano, M.,
Aaraj, N., Latorre, J.I., Solving systems of Boolean multivariate equations with
quantum annealing, Physical Review Research 4(1), 2022.

\bibitem{VDE} van den Essen, A., Polynomial automorphisms and the Jacobian
conjecture. Progress in Mathematics, 190. Birkh\"auser Verlag, Basel, 2000.

\bibitem{Wu} Wu, W.T., On the Decision Problem and the Mechanization of Theorem-Proving
in Elementary Geometry, Sci. Sinica 21 (1978), no. 2, 159--172.
Also in: Automated theorem proving (Denver, Col., 1983), 213--234, Contemp. Math., 29,
Amer. Math. Soc., Providence, RI, 1984.

\end{thebibliography}
\end{document}